\documentclass{vldb}
\usepackage{graphicx}
\usepackage{balance}
\usepackage{amssymb}
\usepackage{amsfonts}
\usepackage{color}
\usepackage{wrapfig}
\usepackage{euscript}
\usepackage{epsfig}
\usepackage{xspace}
\usepackage{relsize}
\usepackage{colortbl}
\usepackage[update,prepend]{epstopdf}
\usepackage{subfigure,color}

\makeatletter
\newif\if@restonecol
\makeatother

\usepackage[noend]{algpseudocode}
\usepackage[linesnumbered,ruled,vlined, noend]{algorithm2e}
\usepackage[utf8]{inputenc}
\usepackage[english]{babel}
\newtheorem{example}{Example}

\newtheorem{theorem}{Theorem}
\newtheorem{lemma}{Lemma}

\newtheorem{definition}{Definition}

\usepackage[noadjust]{cite}
\usepackage{slashbox}
\usepackage[table]{xcolor}
\usepackage[justification=centering]{caption}
\usepackage{multirow}

\usepackage[singlelinecheck=off]{caption}
\usepackage[colorlinks,allcolors=black]{hyperref}

\newcommand{\kwnospace}[1]{{\ensuremath {\mathsf{#1}}}}
\newlength{\figsize} 

\renewcommand{\footnotesize}{\small}

\newcommand{\Break}{\textbf{break}}
\newcommand{\fg} {G^*}

\newcommand{\tc} {TC}
\newcommand{\thop} {2hop_G}
\newcommand{\nb} {N_G}
\newcommand{\p} {p}
\newcommand{\nbf} {N_{G^*}}
\newcommand{\degree} {deg_G}
\newcommand{\btf}{\,\mathbin{\resizebox{0.055in}{!}{\rotatebox[origin=c]{90}{$\Join$}}}}
\newcommand{\btfc}{\textit{butterfly counting}\xspace}

\newcommand{\bs}{\kwnospace{BFC}\textrm{-}\kwnospace{BS}\xspace}
\newcommand{\bsa}{\kwnospace{BFC}\textrm{-}\kwnospace{IBS}\xspace}
\newcommand{\new} {\kwnospace{BFC}\textrm{-}\kwnospace{VP}\xspace}
\newcommand{\newa} {\kwnospace{BFC}\textrm{-}\kwnospace{VP^{++}}\xspace}
\newcommand{\newaw} {\kwnospace{BFC}\textrm{-}\kwnospace{VP^{+}}\xspace}
\newcommand{\newac} {\kwnospace{BFC}\textrm{-}\kwnospace{VPC}\xspace}
\newcommand{\bsap} {\kwnospace{BFC}\textrm{-}\kwnospace{ESap}\xspace}
\newcommand{\newap} {\kwnospace{BFC}\textrm{-}\kwnospace{ESap_{vp^{++}}}\xspace}
\newcommand{\bse}{\kwnospace{BFC}\textrm{-}\kwnospace{EIBS}\xspace}
\newcommand{\newe} {\kwnospace{BFC}\textrm{-}\kwnospace{EVP}\xspace}
\newcommand{\newae} {\kwnospace{BFC}\textrm{-}\kwnospace{EVP^{++}}\xspace}
\newcommand{\newaem} {\kwnospace{BFC}\textrm{-}\kwnospace{EM}\xspace}
\newcommand{\bsf} {baseline butterfly counting algorithm\xspace}
\newcommand{\bsaf} {improved baseline butterfly counting algorithm\xspace}
\newcommand{\newf} {vertex-priority-based butterfly counting algorithm\xspace}
\newcommand{\newaf} {cache-aware butterfly counting algorithm\xspace}

\newcommand{\fin} {f^{-1}}
\newcommand{\cate}{\mathbin{\resizebox{0.098in}{!}{\rotatebox[origin=c]{270}{$\ltimes$}}}}
\usepackage{enumitem}
\setlist{nolistsep}
\usepackage{amsmath}
\makeatletter
\g@addto@macro\normalsize{%
\setlength\abovedisplayskip{-1pt}
\setlength\belowdisplayskip{0pt}
\setlength\abovedisplayshortskip{-1pt}
\setlength\belowdisplayshortskip{0pt}
}
\makeatother

\begin{document}

\title{Efficient Butterfly Counting for Large Bipartite Networks}

\author{
\alignauthor Kai Wang$^{\dagger\ddagger}$, Xuemin
Lin$^{\dagger\ddagger}$, Lu Qin$^\star$, Wenjie Zhang$^\dagger$, Ying Zhang$^\star$
\vspace{0.1cm} \\
       \affaddr{$^\dagger$University of New South Wales, Australia}\ \ \
       \affaddr{$^\ddagger$Zhejiang Lab, China} \ \ \
       \affaddr{$^\star$University of Technology Sydney, Australia} \\
	   \affaddr{{kai.wang@unsw.edu.au, \{lxue,zhangw\}@cse.unsw.edu.au, \{lu.qin, ying.zhang\}@uts.edu.au}}
}

\maketitle

\begin{abstract}
Bipartite networks are of great importance in many real-world applications. In bipartite networks, butterfly (i.e., a complete $2 \times 2$ biclique) is the smallest non-trivial cohesive structure and plays a key role. In this paper, we study the problem of efficient counting the number of butterflies in bipartite networks. The most advanced techniques
are based on enumerating wedges which is the dominant cost of counting butterflies. Nevertheless, the existing algorithms cannot efficiently handle large-scale bipartite networks. This becomes a bottleneck in large-scale applications. In this paper, instead of the existing layer-priority-based techniques, we propose a vertex-priority-based paradigm \new to enumerate much fewer wedges; this leads to a significant improvement of the time complexity of the state-of-the-art algorithms.
In addition, we present cache-aware strategies to further improve the time efficiency while theoretically retaining the time complexity of \new.
Moreover, we also show that our proposed techniques can work efficiently in external and parallel contexts.
Our extensive empirical studies demonstrate that the proposed techniques can speed up the state-of-the-art techniques by up to two orders of magnitude for the real datasets.

\end{abstract}

\section{Introduction}
\label{sct:introduction}

When modelling relationships between two different types of entities, the bipartite network arises naturally as a data model in many real-world applications \cite{borgatti1997network, latapy2008basic}. For example, in online shopping services (e.g., Amazon and Alibaba), the purchase relationships between users and products can be modelled as a bipartite network, where users form one layer, products form the other layer, and the links between users and productions represent purchase records as shown in Figure \ref{fig:example1}. Other examples include author-paper relationships, actor-movie networks, etc.

\begin{figure}[thb]
\begin{centering}
\includegraphics[trim=0 0 0 0,width=0.30\textwidth]{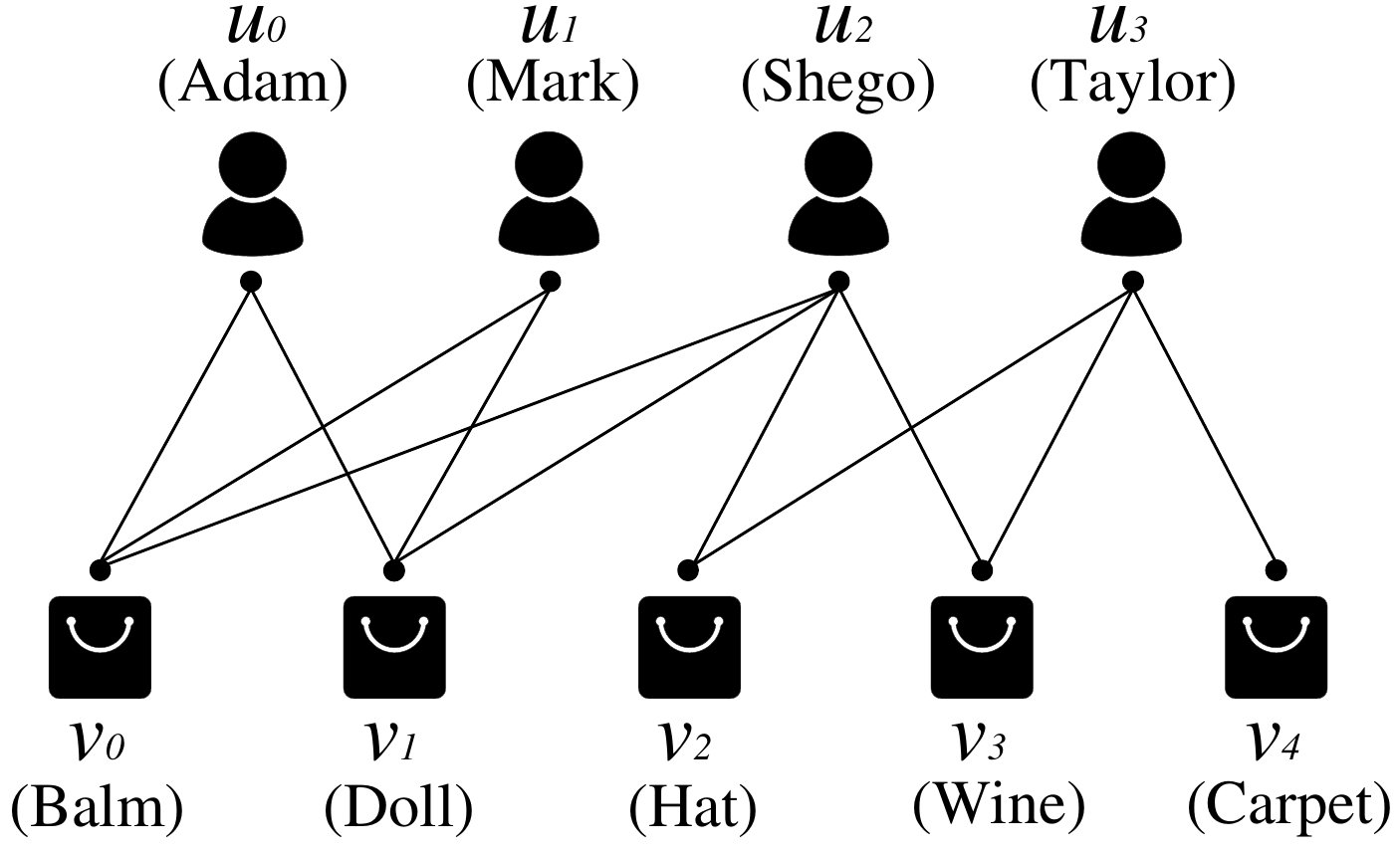}
\vspace*{-3mm}\caption{A bipartite network}
\label{fig:example1}
\vspace*{-1mm}
\end{centering}
\end{figure}

Since network motifs (i.e., repeated sub-graphs) are regarded as basic building blocks of complex networks \cite{milo2002network}, finding and counting motifs of networks is a key to network analysis. In unipartite networks, there are extensive studies on counting and listing triangles (the smallest non-trivial clique) in the literature \cite{chang2017scalable, kolountzakis2012efficient, chu2012triangle, seshadhri2013triadic, stefani2017triest, shun2015multicore, al2018triangle, hu2013massive, latapy2008main, schank2005finding, suri2011counting}. In bipartite networks, {\em butterfly} (i.e., a complete 2 $\times$ 2 biclique) is the simplest bi-clique configuration with equal numbers of vertices of each layer (apart from the trivial single edge configuration) that has drawn reasonable attention recently \cite{robins2004small, sanei2018butterfly, wang2014rectangle, aksoy2017measuring, sariyuce2018peeling, zou2016bitruss}; for instance, Figure \ref{fig:example1} shows the record that Adam and Mark both purchased Balm and Doll forms a butterfly. In this sense, the butterfly can be viewed as an analogue of the triangle in a unipartite graph. Moreover, without butterflies, a bipartite graph will not have any community structure \cite{aksoy2017measuring}.


In this paper, we study the problem of {\em butterfly counting}, that is to compute the number of butterflies in a bipartite graph $G$, denoted by $\btf_G$. The importance of {\em butterfly counting} has been demonstrated in the literature of network analysis and graph theory. Below are some examples.

\noindent
{\em Network measurement.} The {\em bipartite clustering coefficient} \cite{robins2004small, lind2005cycles, aksoy2017measuring, opsahl2013triadic} is a cohesiveness measurement of bipartite networks. Given a bipartite graph $G$, its bipartite clustering coefficient equals $4 \times \btf_G$/$\cate_G$, where $\cate_G$ is the number of caterpillars in $G$ --- the number of three-paths. For example, ($u_0$, $v_0$, $u_1$, $v_1$) in Figure \ref{fig:example1} is a three-path. High bipartite clustering coefficient indicates localized closeness and redundancy in bipartite networks \cite{robins2004small, aksoy2017measuring}; for instance, in user-product networks, bipartite clustering coefficients can be used frequently to analyse the sale status for products in different categories. These statistics can also be used in \texttt{Twitter} network \cite{Twitter} for internet advertising where the \texttt{Twitter} network is the bipartite network consisting of Twitter users and the URLs they mentioned in their postings. Since $\cate_G$ can be easily computed in $O(m)$ time where $m$ is the number of edges in $G$ \cite{aksoy2017measuring}, computing $\btf_G$ becomes a bottleneck in computing the clustering coefficient.


\noindent
{\em Summarizing inter-corporate relations.}
In a director-board network, two directors on the same two boards can be modelled as a butterfly. These butterflies can reflect inter-corporate relations \cite{palmer1983broken, ornstein1984interlocking, ornstein1982interlocking}.
The number of butterflies indicates the extent to which directors re-meet one another on two or more boards. A large butterfly counting number indicates a large number of inter-corporate relations and formal alliances between companies \cite{robins2004small}.

\noindent
{\em Computing $k$-wing in bipartite graphs.}
Counting the number of butterflies for each edge also has applications. For example, it is the first step to compute a $k$-wing \cite{sariyuce2018peeling}  (or $k$-bitruss \cite{zou2016bitruss}) for a given $k$ where $k$-wing is the maximum subgraph of a bipartite graph with each edge in at least $k$ butterflies. Discovering such dense subgraphs is proved useful in many applications e.g., community detection \cite{peng2018efficient, zhang2018discovering, zhang2018efficiently, yixiang2019survey}, word-document clustering \cite{dhillon2001co}, and viral marketing \cite{fain2006sponsored, liu2019, zhang2017engagement, wang2018efficient}. Given a bipartite graph $G$, the proposed algorithms \cite{zou2016bitruss, sariyuce2018peeling} for $k$-wing computation is to first count the number of butterflies on each edge in $G$. After that, the edge with the lowest number of butterflies is iteratively removed from $G$ until all the remaining edges appear in at least $k$ butterflies.

Note that in real applications, butterfly counting may happen not only once in a graph. We may need to conduct such a computation against an arbitrarily specified subgraph. Indeed, the demand of butterfly counting in large networks can be very high.
However, the state-of-the-art algorithms cannot efficiently handle large-scale bipartite networks. As shown in \cite{sanei2018butterfly}, on the \texttt{Tracker} network with $10^8$ edges, their algorithm needs about 9,000 seconds to compute $\btf_G$. Therefore, the study of efficient butterfly counting is imperative to support online large-scale data analysis. Moreover, some applications demand exact butterfly counting in bipartite graphs. For example, in $k$-wing computation, approximate counting does not make sense since the $k$-wing decomposition algorithm in \cite{sariyuce2018peeling} needs to iteratively remove the edges with the lowest number of butterflies; the number has to be exact.



\noindent
{\bf State-of-the-art.} Consider that there can be $O (m^2)$ butterflies in the worst case. Wang et al. in \cite{wang2014rectangle} propose an algorithm to avoid enumerating all the butterflies. It has two steps. At the first step, a layer is randomly selected.  Then, the algorithm iteratively starts from every vertex $u$ in the selected layer, computes the $2$-hop reachable vertices from $u$, and for each $2$-hop reachable vertex $w$, counts the number $n_{uw}$ of times reached from $u$. At the second step, for each $2$-hop reachable vertex $w$ from $u$, we count the number of butterflies containing both $u$ and $w$ as $n_{uw} (n_{uw}-1)/2$. For example, regarding Figure \ref{fig:example1}, if the lower layer is selected, starting from the vertex $v_0$, vertices $v_1$, $v_2$, and $v_3$ are $2$-hop reached $3$ times, $1$ time, and $1$ time, respectively. Thus, there are $C_3^2$ ($=3$) butterflies containing $v_0$ and $v_1$ and no butterfly containing $v_0$ and $v_2$ (or $v_0$ and $v_3$). Iteratively, the algorithm will first use $v_0$ as the start-vertex, then $v_1$, and so on. Then, we add all the counts together; the added counts divided by two is the total number of butterflies.

Observe that the time complexity of the algorithm in \cite{wang2014rectangle} is $O (\sum_{u \in U(G)}\degree(u)^2))$ if the lower layer $L (G)$ of $G$ is chosen to have start-vertices, where $U (G)$ is the upper layer. Sanei et al. in \cite{sanei2018butterfly} chooses a layer $S$ such that $O (\sum_{v \in S}\degree(v)^2))$ is minimized among the two layers.

%

\vspace{0.08cm}
\noindent
{\bf Observation.} In the existing algorithms \cite{wang2014rectangle,sanei2018butterfly},
the dominant cost is at Step 1 that enumerates wedges to compute $2$-hop reachable vertices and their hits. For example, regarding Figure \ref{fig:example1}, we have to traverse $3$ wedges, $(v_0, u_0, v_1)$, $(v_0, u_1, v_1)$, and $(v_0, u_2, v_1)$ to get all the hits from $v_0$ to $v_1$. Here, in the wedge  $(v_0, u_0, v_1)$, we refer $v_0$ as the start-vertex, $u_0$ as the middle-vertex, and $v_1$ as the end-vertex.
Continue with the example in Figure \ref{fig:example1}, using $u_2$ as the middle-vertex, starting from $v_0$, $v_1$, and $v_2$, respectively, we need to traverse totally $6$ wedges.

We observe that the choice of middle-vertices of wedges (i.e., the choice of start-vertices) is a key to improving the efficiency of counting butterflies. For example, consider the graph $G$ with $2,002$ vertices and $3,000$ edges in Figure \ref{fig:ideas}(a). In $G$, $u_0$ is connected with $1,000$ vertices ($v_0$ to $v_{999}$),  $v_{1000}$ is also connected with $1,000$ vertices ($u_1$ to $u_{1000}$), and for $0 \leq i \leq 999$, $v_i$ is connected with $u_{i+1}$. The existing algorithms need to go through $u_0$ (or $v_{1000}$) as the middle-vertex if choosing $L(G)$ (or $U (G)$) to start. Therefore, regardless of whether the upper or the lower layer is selected to start, we have to traverse totally $C_{1000}^2$ ($=499,500$) plus $1,000$ different wedges by the existing algorithms \cite{wang2014rectangle,sanei2018butterfly}.

\begin{figure}[htb]
\begin{centering}
\includegraphics[trim=0 0 0 0,width=0.47\textwidth]{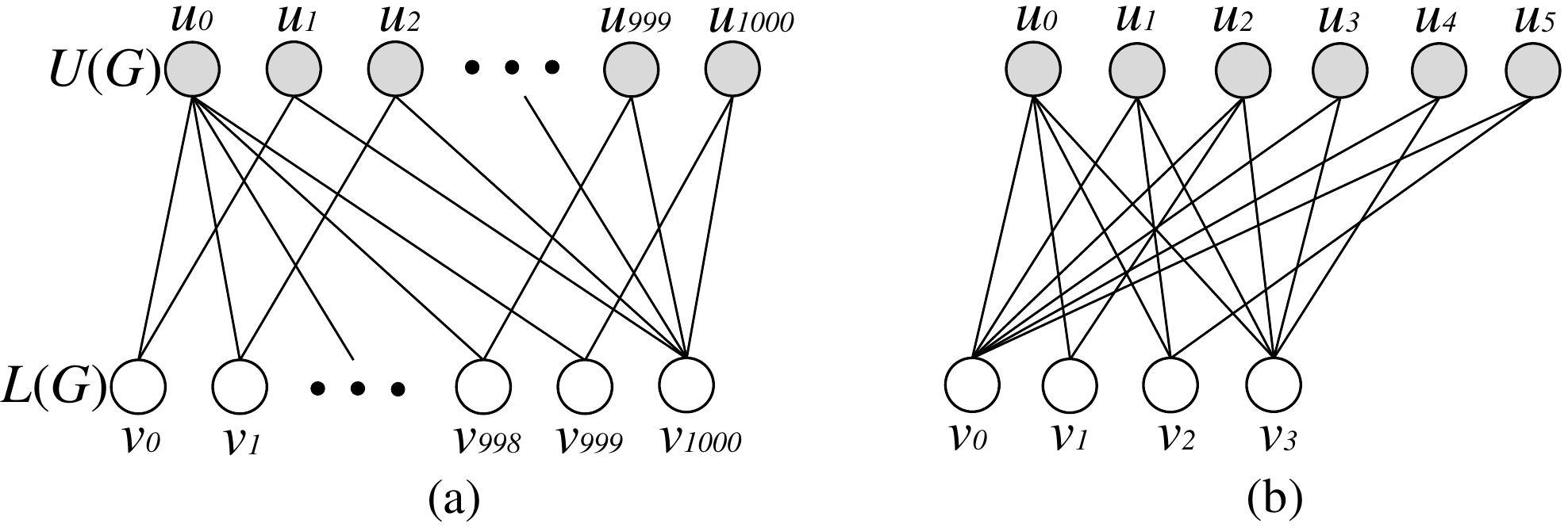}
\vspace*{-3mm}
\caption{Some observations}
\label{fig:ideas}
\end{centering}
\end{figure}

\noindent
{\bf Challenges.}
The main challenges of efficient butterfly counting are twofold.

\begin{enumerate}
\item Using high-degree vertices as middle-vertices of wedges may generate numerous wedges to be scanned. The existing techniques \cite{wang2014rectangle,sanei2018butterfly}, including the layer-priority-based techniques \cite{sanei2018butterfly}, cannot avoid using unnecessary high-degree vertices as middle-vertices as illustrated earlier.   Therefore, it is a challenge to effectively handle high-degree vertices. 
\item Effectively utilizing CPU cache can often reduce the computation dramatically. Therefore, it is also a challenge to utilize CPU cache to speed up the counting of butterflies.
\end{enumerate}

\noindent
{\bf Our approaches.} To address Challenge 1, instead of the existing layer-priority-based algorithm, we propose a \newf\ \new that can effectively handle hub vertices (i.e., high-degree vertices). To avoid over-counting or miss-counting, we propose that for each edge $(u, v)$, the new algorithm \new uses
the vertex with a higher degree as the start-vertex so that the vertex with a lower degree will be used as the middle-vertex. Specifically, the \new algorithm will choose one end vertex of an edge $(u, v)$ as the start-vertex, say $u$, according to its priority. The higher degree, the higher priority; and the ties are broken by vertex ID.  For example, regarding Figure \ref{fig:ideas}(a), the \new algorithm will choose $u_0$ and $v_{1000}$ as start-vertices; consequently, only $2,000$ wedges in total will be scanned by our algorithm compared with $500,500$ different wedges generated by the existing algorithms as illustrated earlier. Once all edges from the starting vertex $u$ are exhausted, \new moves to another edge. This is the main idea of our \new algorithm.

As a result, the time complexity of our \new algorithm is $O(\sum_{(u, v) \in E(G)}min \{\degree(u), \degree(v)\})$ which is in general significantly lower than the time complexity of the state-of-the-art algorithm in \cite{sanei2018butterfly},
$O (\min \{ \sum_{v \in U(G)}\degree(v)^2, \sum_{v \in L(G)}\degree(v)^2)\} )$, considering $\degree(v)^2 = \sum_{(u, v) \in E (G)} \degree(v)$ where $v$ is fixed.
In the \new algorithm, there are $O(n)$ accesses of start-vertices because we need to explore every vertex as a start-vertex only once, $O(m)$ accesses of middle-vertices and $O(\sum_{(u, v) \in E(G)}min\{\degree(u), \degree(v)\})$ accesses of end-vertices in the processed wedges.
Thus, the number of accesses to end-vertices is dominant.
Given that the cache miss latency takes a big part of the memory access time \cite{ailamaki1999dbmss}, improving the CPU cache performance when accessing the end-vertices becomes a key issue.  Our second algorithm, the cache-aware algorithm \newa, aims to improve the CPU cache performance of \new by having high-degree vertices as end-vertices to enhance the locality while retaining the total $O(\sum_{(u, v) \in E(G)}min\{\degree(u), \degree(v)\})$ accesses of end-vertices (thus, retain the time complexity of the \new algorithm). Consequently,  \newa proposes to request the end-vertices to be prioritized in the same way as the start-vertices in the \new algorithm.

%

For example, considering the graph in Figure \ref{fig:ideas}(b), we have $p(v_0) > p(v_3) > p(u_0) > p(v_2) > p(v_1)$ according to their degrees where $p(v)$ denotes the priority of a vertex $v$. In this example, starting from $v_0$ to $v_3$, going through $u_0$, \new needs to process $5$ wedges using $u_0$ as the middle-vertex (i.e., $(v_0, u_0, v_1)$, $(v_0, u_0, v_2)$, $(v_0, u_0, v_3)$, $(v_3, u_0, v_1)$ and $(v_3, u_0, v_2)$), and there are $3$ vertices, $v_1$, $v_2$ and $v_3$ that need to be performed as end-vertices. Note that these are the only $5$ wedges using $u_0$ as the middle-vertex since $p(u_0) > p(v_2) > p(v_1)$. Regarding the same example, \newa also needs to process exactly $5$ wedges with $u_0$ as the middle-vertex, $(v_1, u_0, v_0)$, $(v_1, u_0, v_3)$, $(v_2, u_0, v_0)$, $(v_2, u_0, v_3)$ and $(v_3, u_0, v_0)$; however only $2$ vertices, $v_0$ and $v_3$, are performed as end-vertices.

We also propose the cache-aware projection strategy to improve the cache performance by storing high-priority (more frequently accessed) end-vertices together to reduce the cache-miss \cite{wei2016speedup}. Considering the example in Figure \ref{fig:ideas}(b), \newa will store $v_0$ and $v_3$ together after projection.


\noindent
{\bf Contribution.}
Our principal contributions are summarized as follows.
\begin{itemize}

\item We propose a novel algorithm \new to count the butterflies that significantly reduce the time complexities of the existing algorithms in both theory and practice.


\item We propose a novel \newaf\ \newa by adopting cache-aware strategies to \new. The \newa algorithm achieves better CPU cache performance than \new.


\item We can replace the exact counting algorithm in the approximate algorithm \cite{sanei2018butterfly} by our exact counting algorithm
 for a speedup.

\item By extending our framework, we present an external-memory algorithm and a parallel algorithm for butterfly counting.

\item We conduct extensive experiments on real bipartite networks. The result shows that our proposed algorithms \new and \newa outperform the state-of-the-art algorithms by up to $2$ orders of magnitude. For instance, the \newa algorithm can count $10^{12}$ butterflies in $50$ seconds on \texttt{Tracker} dataset with $10^8$ edges, while the state-of-the-art butterfly counting algorithm \cite{sanei2018butterfly} runs about $9,000$ seconds.

\end{itemize}


\noindent \textbf{Organization.} The rest of the paper is organized as follows.
The related work follows immediately.
Section~\ref{sct:preliminaries} presents the problem definition.
Section~\ref{sct:benchmark1} introduces the existing algorithms \bs and \bsa. The \new algorithm is presented in Section~\ref{sct:new}. Section~\ref{sct:newa} explores cache-awareness. Section~\ref{sct:extension} extends our algorithm to count butterflies against each edge, the parallel execution of our proposed algorithms and the external memory solution. Section~\ref{sct:experiment} reports experimental results. 
Section~\ref{sct:conclusion} concludes the paper.

\vspace{0.1cm}
\noindent
{\bf Related Work.}

\vspace{0.1cm}
\noindent
{\em Motif counting in unipartite networks.} Triangle is the smallest non-trivial cohesive structure and there are extensive studies on counting triangles in the literature  \cite{chang2017scalable, becchetti2008efficient, schank2005finding, kolountzakis2012efficient, chu2012triangle, seshadhri2013triadic, stefani2017triest, shun2015multicore, al2018triangle, hu2013massive, latapy2008main, itai1978finding, schank2005finding, suri2011counting}. However, the butterfly counting is inherently different from the triangle counting for two reasons, 1) the number of butterflies may be significantly larger than that of triangles ($O(m^2)$ vs $O(m^{1.5})$ in the worst case), and 2) the structures are different ($4$-hops' circle vs $3$-hops' circle).
Thus, the existing triangle counting techniques are not applicable to efficient butterfly counting because the existing techniques for counting triangles (e.g., \cite{chu2012triangle, shun2015multicore}) are based on enumerating all triangles and the enumeration is not affordable in counting butterflies due to the quadratic number  $O (m^2)$ of butterflies in the worst case.

There are also some studies \cite{jain2017fast, pinar2017escape, jha2015path} focusing on the other cohesive structures such as $4$-vertices and $5$-vertices, these techniques also cannot be used to solve our problem. In \cite{alon1997finding}, the authors propose generic matrix-multiplication based algorithm for counting the cycles of length $k$ ($3 \leq k \leq 7$) in $O(n^{2.376})$ time and $O (n^2)$ space. While the algorithm in \cite{alon1997finding} can be used to solve our problem, it cannot process large graphs due to its space and time complexity. As shown in \cite{wang2014rectangle}, the algorithm in \cite{wang2014rectangle} has a significant improvement over \cite{alon1997finding}, while our algorithm significantly improves \cite{wang2014rectangle}.


\vspace{0.1cm}
\noindent
{\em Bipartite Networks.} Some studies are conducted towards motifs such as $3 \times 3$ biclique \cite{borgatti1997network} and $4$-path \cite{opsahl2013triadic}. These structures are different from the butterfly thus these works also cannot be used to solve the butterfly counting problem.
As mentioned earlier, the study in this paper aims to improve the recent works in \cite{wang2014rectangle, sanei2018butterfly}.

\vspace{0.1cm}
\noindent
{\em Graph ordering.} There are some studies on specific graph algorithms using graph ordering. Then et al. \cite{then2014more} optimize BFS algorithms. Park et al. \cite{park2004optimizing} improve the CPU cache performance of many classic graph algorithms such as Bellman-Fold and Prim. The authors in \cite{han2018speeding} present a suite of approaches to accelerate set intersections in graph algorithms. Since these techniques are specific to the problems studied, they are not applicable to butterfly counting.

In the literature, there are also recent works studying general graph ordering methods to speed up graph algorithms \cite{wei2016speedup, blandford2003compact, chierichetti2009compressing, dhulipala2016compressing, kang2011beyond, auroux2015reordering, boldi2009permuting, boldi2011layered, yu2013more}. In the experiments, we show that our cache-aware techniques outperform the state-of-the-art technique \cite{wei2016speedup}; that is, our cache-aware strategy is more suitable for counting butterflies.

\section{Problem Definition}
\label{sct:preliminaries}

In this section, we formally introduce the notations and definitions. Mathematical notations used throughout this paper are summarized in Table~\ref{tb:notations}.

\begin{table}[htb]
\footnotesize
\caption{The summary of notations}
\vspace{-3mm}
  \centering
    \begin{tabular}{|c|l|}
      \hline
      \cellcolor{gray!25}\textbf{Notation} & \cellcolor{gray!25}\textbf{Definition}             \\ \hline

      $G$   &  a bipartite graph \\ \hline
      $V(G)$   &  the vertex set of $G$ \\ \hline
      $E(G)$   & the edge set of $G$ \\ \hline
      $U(G),L(G)$   &  a vertex layer of $G$ \\ \hline
      $u, v, w, x$  & a vertex in the bipartite graph \\ \hline
      $e, (u, v)$  & an edge in the bipartite graph \\ \hline
      $(u, v, w)$  & a wedge formed by $u$, $v$, $w$ \\ \hline
      $[u, v, w, x]$  &  a butterfly formed by $u$, $v$, $w$, $x$\\ \hline
      $deg_G(u)$   & the degree of vertex $u$ \\ \hline
      $\p(u)$   & the priority of vertex $u$ \\ \hline
      $\nb(u)$   & the set of neighbors of vertex $u$ \\ \hline
      $\thop(u)$   & the set of two-hop neighbors of vertex $u$ \\ \hline
      $\btf_e$  & the number of butterflies containing an edge $e$ \\ \hline
      $\btf_G$  & the number of butterflies in $G$ \\ \hline
      $n, m$  & the number of vertices and edges in $G$ ($m > n$) \\ \hline
    \end{tabular}
\label{tb:notations}
\end{table}


Our problem is defined over an undirected bipartite graph $G(V=(U, L), E)$, where $U(G)$ denotes the set of vertices in the upper layer, $L(G)$ denotes the set of vertices in the lower layer, $U(G) \cap L(G) = \emptyset$, $V(G) = U(G) \cup L(G)$ denotes the vertex set, and $E(G) \subseteq U(G) \times L(G)$ denotes the edge set.  We use $n$ and $m$ to denote the number of vertices and edges in $G$, respectively, and we assume $m > n$.
In addition, we use $r$ and $l$ to denote the number of vertices in $U(G)$ and $L(G)$, respectively.
An edge between two vertices $u$ and $v$ in $G$ is denoted as $(u, v)$ or $(v, u)$. The set of neighbors of a vertex $u$ in $G$ is denoted as $\nb(u) = \{ v\in V(G) \mid (u, v) \in E(G) \} $, and the degree of $u$ is denoted as $deg_G(u) = |N_G(u)|$. The set of $2$-hop neighbors of $u$ (i.e., the set of vertices which are exactly two edges away from $u$) is denoted as $\thop(u)$. Each vertex $u$ has a unique id and we assume for every pair of vertices $u \in U(G)$ and $v \in L(G)$, $u.id > v.id$. 


\begin{definition}[Wedge]Given a bipartite graph $G(V, E)$ and vertices $u$, $v$, $w \in V(G)$. A path starting from $u$, going through $v$ and ending at $w$ is called a wedge which is denoted as $(u, v, w)$. For a wedge $(u, v, w)$, we call $u$ the start-vertex, $v$ the middle-vertex and $w$ the end-vertex.
\end{definition}


%

\begin{definition}[Butterfly]Given a bipartite graph $G$ and the four vertices $u, v, w, x \in V(G)$ where $u, w \in U(G)$ and $v, x \in L(G)$, a butterfly induced by the vertices $u, v, w, x$ is a (2,2)-biclique of $G$; that is, $u$ and $w$ are all connected to $v$ and $x$, respectively, by edges.
\end{definition}
A butterfly induced by vertices $u, v, w, x$ is denoted as $[u, v, w, x]$. We denote the number of butterflies containing a vertex $u$ as $\btf_u$, the number of butterflies containing an edge $e$ as $\btf_e$ and the number of butterflies in $G$ as $\btf_G$.


\noindent
\textbf{Problem Statement. }
Given a bipartite graph $G(V,E)$, our \btfc problem is to compute $\btf_G$.



\section{Existing Solutions}
\label{sct:benchmark1}

In this section, we briefly discuss the two existing algorithms, the \bsf \ \bs \cite{wang2014rectangle} and the \bsaf \ \bsa \cite{sanei2018butterfly}. As discussed earlier, both algorithms are based on enumerating
wedges. The following Lemma \ref{lemma:existing} \cite{wang2014rectangle}
is a key to the two algorithms.



\begin{lemma}
\label{lemma:existing}
Given a bipartite graph $G(V, E)$ and a vertex $u \in G$, we have the following equations:\\
\begin{align}
& \btf_u = \sum_{w \in \thop(u)}\binom{|\nb(u) \cap \nb(w)|}{2}\\
& \btf_G = \frac{1}{2}\sum_{u \in U(G)}\btf_u = \frac{1}{2}\sum_{v \in L(G)}\btf_v
\end{align}
\end{lemma}

In fact, \bsa has the same framework as \bs and improves \bs in two aspects: (1) pre-choosing the layer of start-vertices to achieve a lower time complexity; (2) using a hash map to speed up the implementation.
The details of the \bsa algorithm are shown in Algorithm \ref{algo:benchmark}.

\begin{algorithm}[hbt]
\small
\DontPrintSemicolon
\KwIn{$G(V = (U, L), E)$: the input bipartite graph}
\KwOut{$\btf_G$}
$\btf_G \gets 0$\;
$S \gets U(G)$\;
 \If {$\sum_{u \in U(G)}\degree(u)^2 < \sum_{v \in L(G)}\degree(v)^2$} {
    $S \gets L(G)$\;
}
\ForEach{$u \in S$} {
    initialize hashmap $count\_wedge$ with zero\;
    \ForEach{$v \in \nb(u)$} {
        \ForEach{$w \in \nb(v): w.id > u.id$} {
            $count\_wedge(w) \gets count\_wedge(w) + 1$\;
        }
    }
    \ForEach{$w \in count\_wedge$} {
        \If {$count\_wedge(w) > 1$} {
            $\btf_G \gets \btf_G + \binom{count\_wedge(w)}{2}$\;
        }
    }
}
\Return{$\btf_G$}\;
\caption{{\sc \bsa}}
\label{algo:benchmark}
\end{algorithm}

Note that to avoid counting a butterfly twice, for each middle-vertex $v \in \nb(u)$ and the corresponding end-vertex $w \in \nb(v)$, \bsa processes the wedge $(u, v, w)$ only if $w.id > u.id$; consequently, in Algorithm \ref{algo:benchmark} we do not need to use the factor $\frac{1}{2}$ in Equation 2 of Lemma \ref{lemma:existing}.

As shown, the time complexity of \bs is $O(\sum_{v \in L(G)}\degree(v)^2)$ if starting from the layer $U(G)$, while the time complexity of \bsa is $O(min\{\sum_{u \in U(G)}\degree(u)^2, \sum_{v \in L(G)}\degree(v)^2\})$.


\section{Algorithm by Vertex Priority }
\label{sct:new}

\vspace{0.3cm}
\begin{figure}[hbt]
\begin{centering}
\includegraphics[trim=0 10 0 15,width=0.26\textwidth]{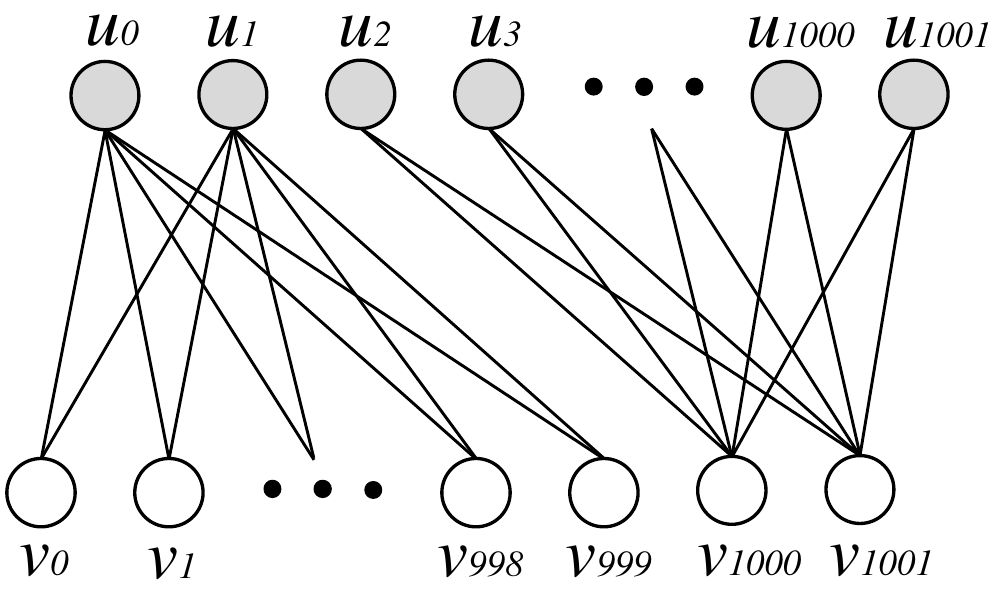}
\vspace*{-2mm}\caption{A bipartite graph containing hub vertices $u_0, u_1, v_{1000}$ and $v_{1001}$.}
\label{fig:hub_vertices}
\vspace*{-1mm}
\end{centering}
\end{figure}

In \bs and \bsa, the time complexity is related to the total number of $2$-hop neighbors visited (i.e.,the total number of wedges processed). When starting from one vertex layer (e.g., $U(G)$), the number of processed wedges is decided by the sum of degree squares of middle-vertices in the other layer (e.g., $\sum_{v \in L(G)}\degree(v)^2$). If all the vertices with lower-degrees are distributed in one vertex layer as middle-vertices, \bsa can just start from the vertices in the other layer and obtain a much lower computation cost.
However, when there are vertices with high-degrees (i.e., {\em hub vertices}) exist in both layers, which is not uncommon in real datasets (e.g., \texttt{Tracker} dataset), choosing which layer to start cannot achieve a better performance. For example, consider the graph $G$ with $2,002$ vertices and $4,000$ edges in Figure \ref{fig:hub_vertices}, where $u_0$ and $u_1$ are connected with $1,000$ vertices ($v_0$ to $v_{999}$),  $v_{1000}$ and $v_{1001}$ are also connected with $1,000$ vertices ($u_2$ to $u_{1001}$). In this example, choosing either of the two layers still needs to go through hub vertices,  $u_0, u_1 \in U(G)$ or $v_{1000}, v_{1001} \in L(G)$. 

\vspace{0.1cm}
\noindent
{\textbf{Optimization strategy.} Clearly, $[u_0, v_0, u_1, v_1]$ in Figure \ref{fig:hub_vertices} can be constructed in two ways: 1) by the wedges $(u_0, v_0, u_1)$ and $(u_0, v_1, u_1)$, or 2) by the wedges $(v_0, u_0, v_1)$ and $(v_0, u_1, v_1)$.  Consequently, a hub vertex (e.g., $u_0$ in Figure \ref{fig:hub_vertices}) may not always necessary to become a middle-vertex in a wedge for the construction of a butterfly. Thus, it is possible to design an algorithm which can avoid using hub vertices unnecessarily as middle-vertices. To achieve this objective, we introduce the \newf \xspace \new which runs in a vertex level (i.e., choosing which vertex to be processed as the start-vertex) rather than a layer level (i.e., choosing which vertex-layer to be processed as the start-layer). The time complexity of \new is $O(\sum_{(u, v) \in E(G)}min\{\degree(u), \degree(v)\})$.

Given a graph $G$, the \new algorithm first assigns a {\em priority} to each vertex $u \in V(G)$ which is defined as follows.

\begin{definition}[Priority]Given a bipartite graph $G(V, E)$, for a vertex $u \in V(G)$, the priority $p(u)$ is an integer where $p(u) \in [1, |V(G)|]$. For two vertices $u, v \in V(G)$, $p(u) > p(v)$ if\\
\vspace*{-2mm}
\begin{itemize}
\item $\degree(u) > \degree(v)$, or \\
\vspace*{-2mm}
\item $\degree(u) = \degree(v)$, $u.id > v.id$.\\
\end{itemize}
\label{definition:priority}
\end{definition}
\vspace*{-4mm}

Given the priority, a butterfly can always be constructed from two wedges $(u,v,w)$ and $(u,x,w)$ where the start-vertex $u$ has a higher priority than the middle-vertices $v$ and $x$. This is because we can always find a vertex which has the highest priority and connects to two vertices with lower priorities in a butterfly.


Based on the above observation, the \new algorithm can get all the butterflies by only processing the wedges where the priorities of start-vertices are higher than the priorities of middle-vertices. In this way, the algorithm \new will avoid processing
the wedges where middle-vertices have higher priorities than start-vertices (e.g.,
$(v_0, u_0, v_1)$ in Figure \ref{fig:hub_vertices}). In addition, in order to avoid duplicate counting, another constraint should also be satisfied in \new: \new only processes the wedges where start-vertices have higher priorities than end-vertices. To avoid processing unecessary wedges in the implementation, we sort the neighbors of vertices in ascending order of their priorities. Then we can early terminate the processing once we meet an end-vertex which has higher priority than the start-vertex (or meet a middle-vertex which has higher priority than the start-vertex).
The details of the \new algorithm are shown in Algorithm \ref{algo:new}. 

\begin{algorithm}[th]
\small
\DontPrintSemicolon
\KwIn{$G(V = (U, L), E)$: the input bipartite graph}
\KwOut{$\btf_G$} 
Compute $p(u)$ for each $u \in V(G)$  // Definition \ref{definition:priority} \;
Sort $N(u)$ for each $u \in V(G)$ according to their priorities\;
$\btf_G \gets 0$\;
\ForEach{ $u \in V(G)$ } {
    initialize hashmap $count\_wedge$ with zero\;
    \ForEach{$v \in \nb(u): p(v) < p(u)$} {
        \ForEach{$w \in \nb(v): p(w) < p(u)$} {
            $count\_wedge(w) \gets count\_wedge(w) + 1$\;
        }
    }
    \ForEach{$w \in count\_wedge$} {
        \If {$count\_wedge(w) > 1$} {
            $\btf_G \gets \btf_G + \binom{count\_wedge(w)}{2}$\;
        }
    }
}
\Return{$\btf_G$}\; 
\caption{{\sc \new}}
\label{algo:new}
\end{algorithm}

Given a bipartite graph $G$, \new first assigns a priority to each vertex $u \in V(G)$ according to Definition \ref{definition:priority} and sort the neighbors of $u$.
After that, \new processes the wedges from each start-vertex $u \in V(G)$ and initializes the hashmap $count\_wedge$ with zero. For each middle-vertex $v \in \nb(u)$, we process $v$ if $p(v) < p(u)$ according to the processing rule. Then, to avoid duplicate counting, we only process $w \in \nb(v)$ with $p(w) < p(u)$. After running lines 4 - 8, we get $|\nb(u) \cap \nb(w)|$ (i.e., $count\_wedge(w)$) for the start-vertex $u$ and the end-vertex $w$. Then, according to Lemma \ref{lemma:existing}, \new computes $\btf_G$.  Finally, we return $\btf_G$.

\noindent
{\textbf{Analysis of the \new algorithm.}
Below we show the correctness and the time complexity of \new.

\begin{figure}[htb]
\begin{centering}
\includegraphics[trim=0 0 0 0,width=0.42\textwidth]{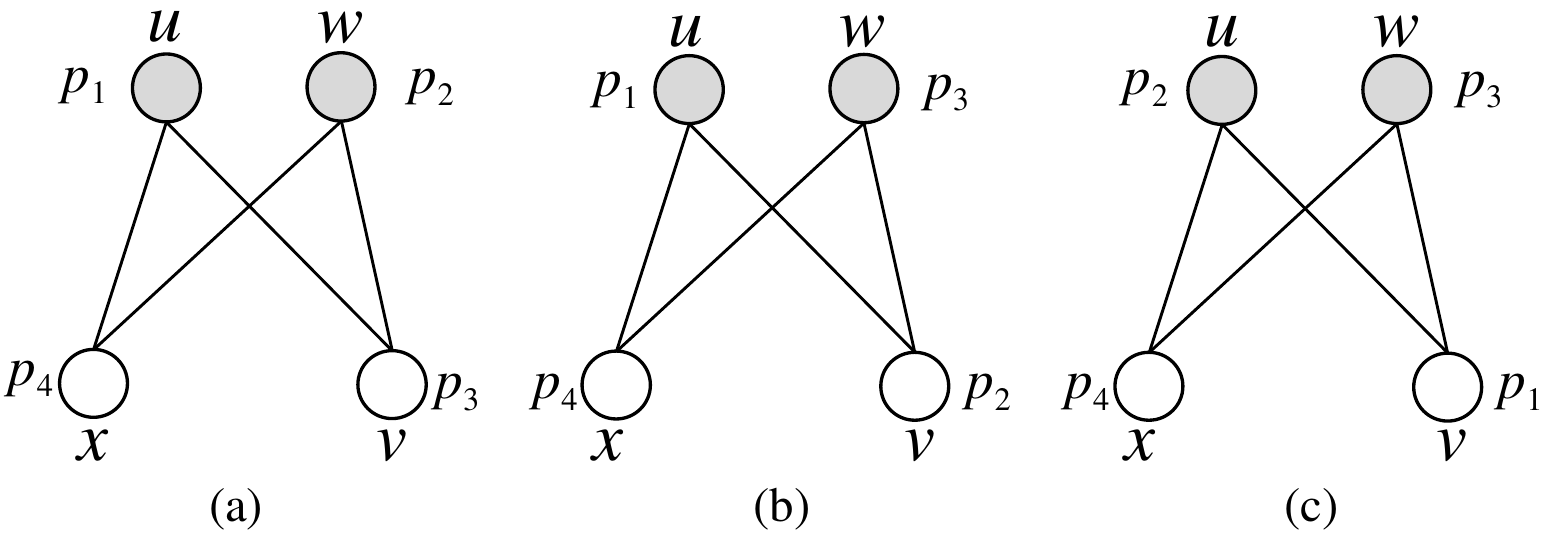}
\vspace*{-3mm}\caption{Assume $p_4 > p_3 > p_2 > p_1$}
\label{fig:proof}
\vspace*{-1mm}
\end{centering}
\end{figure}

\begin{theorem}
\label{theorem:newc}
The \new algorithm correctly solves the \btfc problem.
\end{theorem}

\begin{proof}
\label{proof:newc}
We prove that \new correctly computes $\btf_G$ for a bipartite graph $G$. A butterfly can always be constructed from two different wedges with the same start-vertex and the same end-vertex. Thus, we only need to prove that each butterfly in $G$ will be counted exactly once by \new. Given a butterfly $[x, u, v, w]$, we assume $x$ has the highest priority. The vertex priority distribution must be one of the three situations as shown in Figure \ref{fig:proof} (the other situations can be transformed into the above by a symmetric conversion), where $p_i$ is the corresponding vertex priority. Regarding the case in Figure \ref{fig:proof}(a), \ref{fig:proof}(b), or \ref{fig:proof}(c), \new only counts the butterfly $[x, u, v, w]$ once from the wedges $(x, u, v)$ and $(x, w, v)$. Thus, we can prove that \new correctly solves the \btfc problem. 
\end{proof}

\begin{theorem}
\label{theorem:newtc}
The time complexity of \new is $O(\sum_{(u, v) \in E(G)}min\{\degree(u), \degree(v)\})$.

\end{theorem}
\begin{proof}
The Algorithm \ref{algo:new} has two phases: initializing in the first phase and computing $\btf_G$ in the second phase. The time complexity of the first phase is $O(n + m)$. Firstly, we need $O(m)$ to get the degrees of vertices and $O(n)$ time to get the priorities by sorting the vertices using bin sort \cite{khaouid2015k}. Secondly, we need $O(m)$ time to sort the neighbors of vertices in ascending order of their priorities. To achieve this, we generate a new empty neighbor list $T(u)$ for each vertex $u$. Then we process the vertex with lower priority first and for each vertex $u$ and its neighbor $v$, we put $u$ into $T(v)$. Finally, the neighbors of vertices are ordered in $T$. The time cost of the second phase is related to the number of processed wedges and each wedge needs $O(1)$ time to process. In \new, we only need to process the wedges where the degrees of middle-vertices are lower or equal than the degrees of start-vertices based on the processing rule of \new and Definition \ref{definition:priority}. Considering an edge $(u, v) \in E(G)$ connecting a start-vertex $u$ and a middle-vertex $v$, \new needs to process $O(\degree(v))$ end-vertices from $(u, v)$. That is, for each edge $(u, v) \in E(G)$, \new needs to process $O(min\{\degree(u), \degree(v)\})$ wedges since the middle-vertex has a lower or equal degree than the start-vertex in a processed wedge. In total, \new needs to process $O(\sum_{(u, v) \in E(G)}min\{\degree(u), \degree(v)\})$ wedges. Therefore, the time complexity of \new is $O(\sum_{(u, v) \in E(G)}min\{\degree(u), \degree(v)\})$.
\end{proof}

\begin{theorem}
\label{theorem:newsc}
The space complexity of \new is $O(m)$.
\end{theorem}

\begin{proof}
This theorem is immediate.
\end{proof}

\begin{lemma}
\label{lemma:tmc}
Given a bipartite graph $G$, we have the following equation:\\
\begin{align}
\label{eq:tmc}
\sum_{(u, v) \in E(G)}min\{\degree(u), \degree(v)\} \leq \nonumber \\
min\{\sum_{u \in U(G)}\degree(u)^2, \sum_{v \in L(G)}\degree(v)^2\}
\end{align}
The equality happens if and only if one of the following two conditions is satisfied: (1) for every edge $(u, v) \in E(G)$ and $u \in U(G)$, $\degree(u) \leq \degree(v)$; (2) for every edge $(u, v) \in E(G)$ and $u \in U(G)$, $\degree(v) \leq \degree(u)$.
\end{lemma}

\begin{proof}
Given a bipartite graph $G$, since there are $\degree(u)$ edges attached to a vertex $u$, we can get that:

\begin{align}
\label{eq:tmcp1}
\sum_{u \in U(G)}\degree(u)^2 = \sum_{(u, v) \in E(G), u \in U(G)}\degree(u) \nonumber \\
\geq \sum_{(u, v) \in E(G)}min\{\degree(u), \degree(v)\}
\end{align}

Similarly,

\begin{align}
\label{eq:tmcp2}
\sum_{v \in L(G)}\degree(v)^2 = \sum_{(u, v) \in E(G), u \in U(G)}\degree(v) \nonumber \\
\geq \sum_{(u, v) \in E(G)}min\{\degree(u), \degree(v)\}
\end{align}

Thus, we can prove that Equation \ref{eq:tmc} holds. The condition of equality can be easily proved by contradiction which is omitted here.
\end{proof}

From Lemma \ref{lemma:tmc}, we can get that \new improves the time complexity of \bsa. Now we illustrate how \new efficiently handles the hub-vertices compared with \bsa using the following example.

\begin{example}
\label{exp:new}
Consider the bipartite graph $G$ in Figure \ref{fig:hub_vertices}.

\new first assigns a priority to each vertex in $G$ where $p(u_1) > p(u_0) > p(v_{1001}) > p(v_{1000}) > p(u_{1001}) > p(u_{1000}) >... > p(v_1) > p(v_0)$. Starting from $u_1$, \new needs to process $1,000$ wedges ending at $u_0$. Similarly, starting from $v_{1001}$, \new needs to process $1,000$ wedges ending at $v_{1000}$. No other wedges need to be processed by \new. In total, \new needs to process 2,000 wedges.

\bsa processes each vertex $u \in U(G)$ as start-vertex. Starting from $u_0$, \bsa needs to process 1,000 wedges ending at $u_1$. Starting from $u_1$, no wedges need to be processed. In addition, starting from the vertices in $\{u_2, u_3, ..., u_{1001}\}$, \bsa needs to process $999,000$ wedges. In total, \bsa needs to process $1,000,000$ wedges.



\end{example}

\section{Cache-aware Techniques}
\label{sct:newa}


As discussed in Section \ref{sct:introduction}, below is the breakdown of memory accesses to vertices required when processing the wedges: $O(n)$ accesses of start-vertices, $O(m)$ accesses of middle-vertices, and $O(\sum_{(u, v) \in E(G)}min\{\degree(u), \degree(v)\})$ accesses of end-vertices. Thus, the total access of end-vertices is dominant. For example, by running the \new algorithm on \texttt{Tracker} dataset, there are about $6 \times 10^ 9$ accesses of end-vertices while the accesses of start-vertices and middle-vertices are only $4 \times 10^ 7$ and $2 \times 10^ 8$, respectively. Since the cache miss latency takes a big part of the memory access time \cite{ailamaki1999dbmss}, we try to improve the CPU cache performance when accessing the end-vertices.
\begin{figure}[ht]
\begin{centering}
\hspace*{-1mm}
\includegraphics[trim=0 10 10 -10,width=0.45\textwidth]{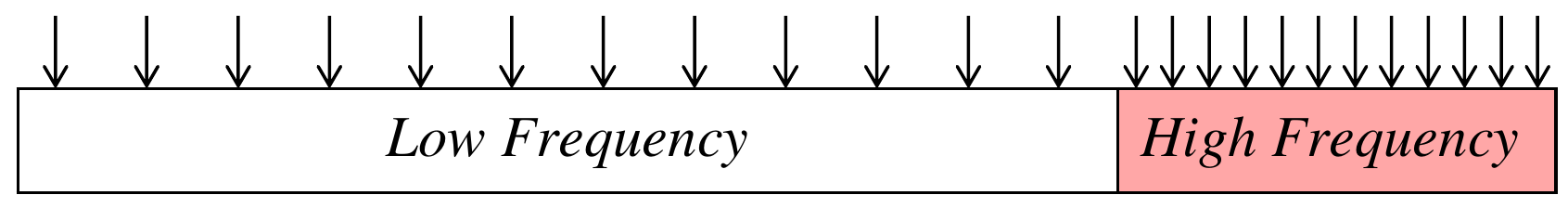}
\vspace*{-1mm}\caption{The buffer {\em BF}}
\label{fig:latency}
\vspace*{-1mm}
\end{centering}
\end{figure}

Because the CPU cache is hard to control in algorithms, a general approach to improve the CPU cache performance is storing frequently accessed vertices together. Suppose there is a buffer {\em BF} and {\em BF}  is partitioned into a {\em low-frequency area LFA} and a {\em high-frequency area  HFA} as shown in Figure \ref{fig:latency}. The vertices are stored in {\em BF} and only a limited number of vertices are stored in {\em HFA}.  For an access of the end-vertex $w$, we compute $miss(w)$ by the following equation:

\begin{equation}
miss(w)=\begin{cases}
1, & \text{iff.}\ w \in LFA, \\
0, & \text{iff.}\ w \in HFA.
\end{cases}
\end{equation}

We want to minimize $F$ which is computed by:

\begin{equation}
\label{equation:hit}
F = \sum_{(u,v,w) \in W}miss(w)\\
\end{equation}

Here, $W$ is the set of processed wedges of an algorithm.

Since $F$ can only be derived after finishing the algorithm, the minimum value of $F$ cannot be pre-computed. We present two strategies which aim to decrease $F$:
\begin{itemize}
\item
Cache-aware wedge processing which performs more high-priority vertices as end-vertices, while retaining the total number of accesses of end-vertices (thus, the same time complexity of \new). Doing this will enhance the access locality. \\
\vspace*{-2mm}
\item
Cache-aware graph projection which stores high-priority vertices together in {\em HFA}. \\
\vspace*{-2mm}

\end{itemize}

\subsection{Cache-aware Wedge Processing}
\label{sct:newaw}
\begin{figure}[hbt]
\begin{centering}
\subfigure[\new]{
\includegraphics[trim=0 0 0 0,clip,width=0.22\textwidth]{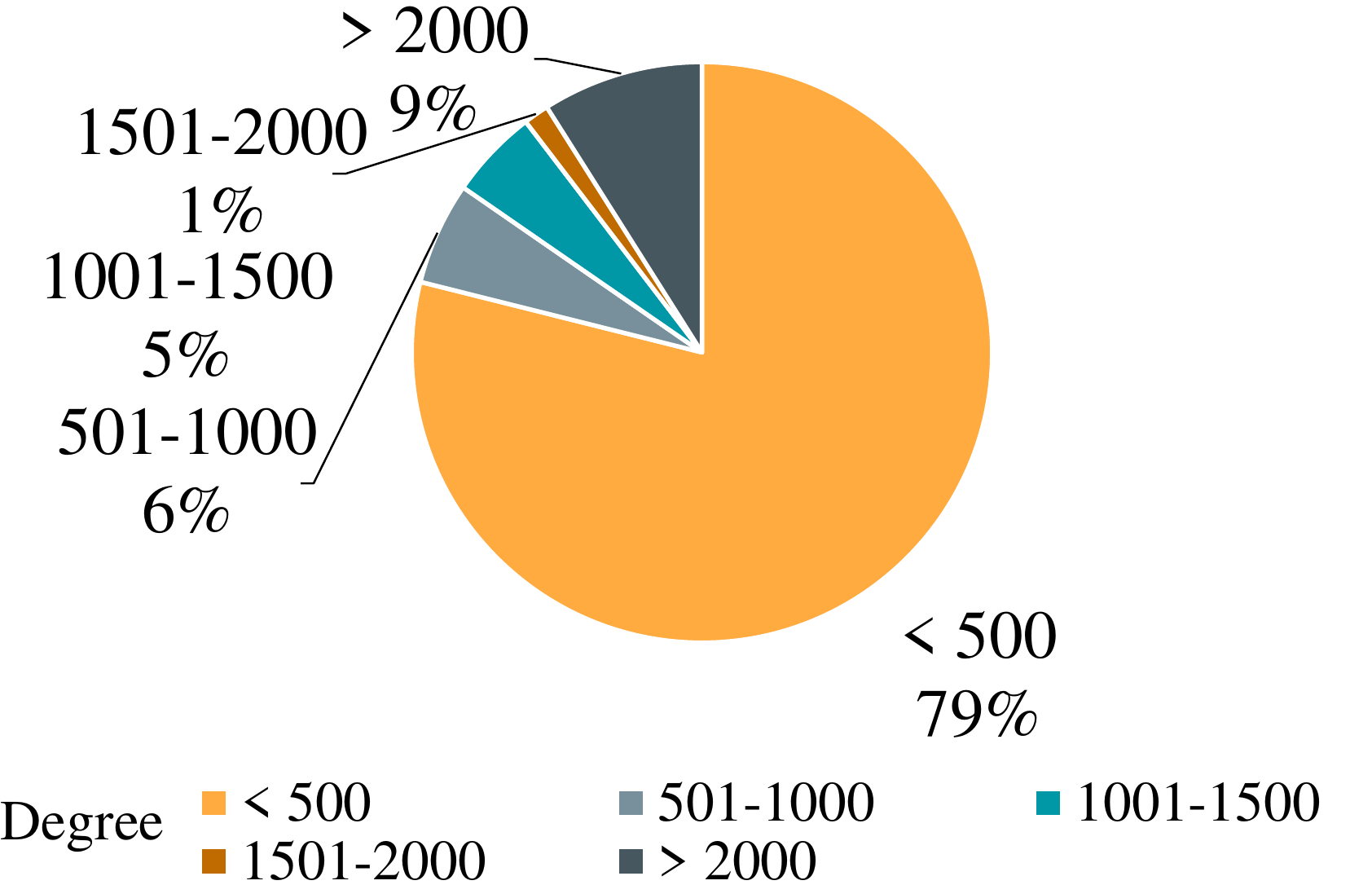}
\label{fig:end1}
}
\vspace{4mm}\subfigure[\newaw]{
\includegraphics[trim=0 0 0 0,clip,width=0.22\textwidth]{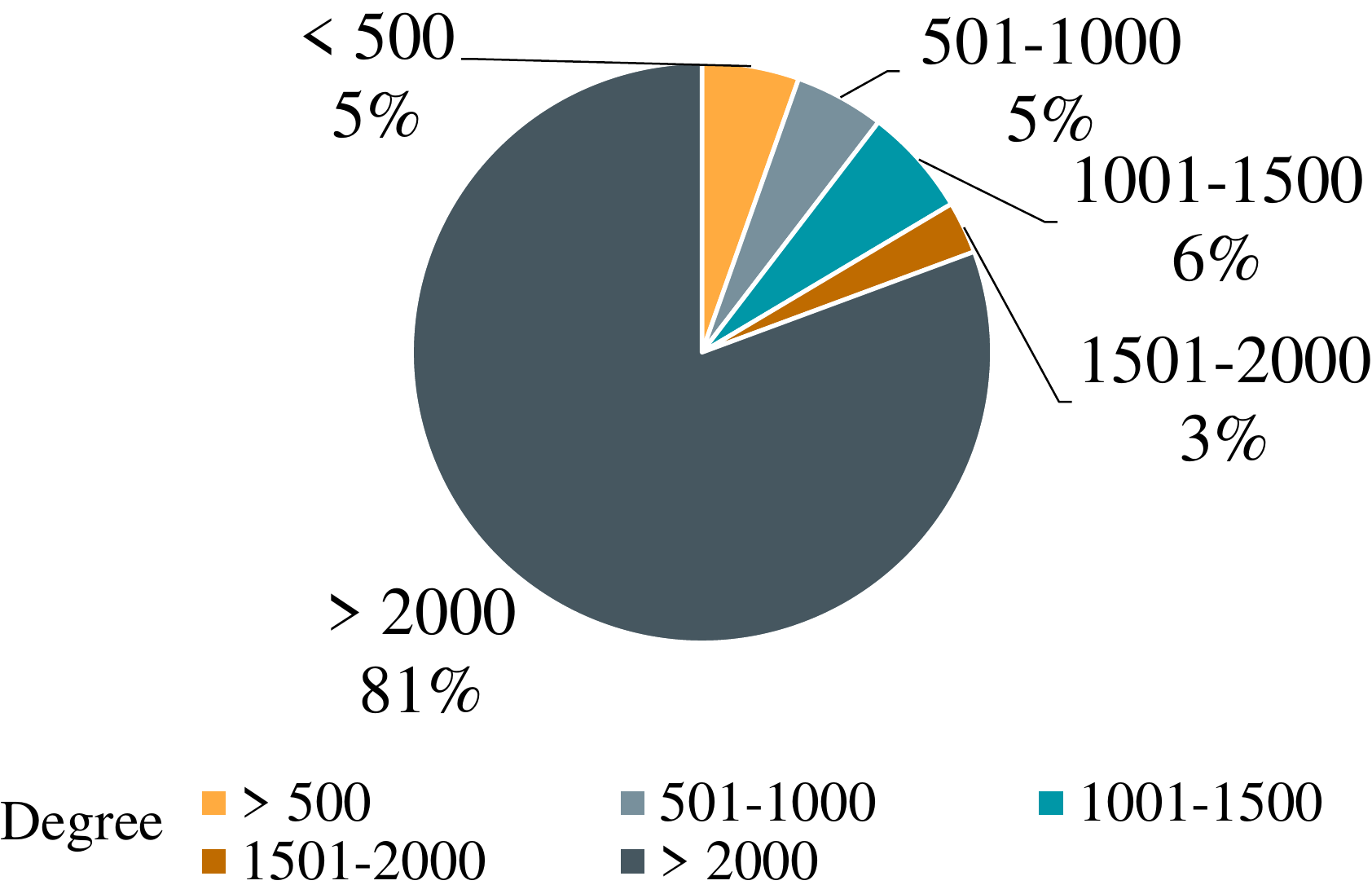}
\label{fig:end2}
}
\vspace*{-8mm}\caption{The degree distribution of the end-vertex-accesses on \texttt{Tracker}}
\label{fig:end}
\end{centering}
\end{figure}
%

\noindent
{\textbf{Issues in wedge processing of \new.}} In \new, the processing rule restricts the priorities of end-vertices should be lower than the priorities of start-vertices in the processed wedges. Because of that, the accesses of end-vertices exhibit bad locality (i.e., not clustered in memory). For example, by counting the accesses of end-vertices over \texttt{Tracker} dataset, as shown in Figure \ref{fig:end1}, $79\%$ of total accesses are accesses of low-degree vertices (i.e., degree $< 500$) while the percentage of high-degree vertices (i.e., degree $> 2000$) accesses is only  $9\%$ in \new. Since the locality of accesses is a key aspect of improving the CPU cache performance, we explore whether the locality of end-vertex-accesses can be improved. With the total access of end-vertices remaining unchanged, we hope the algorithm can access more high-degree vertices as end-vertices. In that manner, the algorithm will have more chance to request the same memory location repeatedly and the accesses of {\em HFA} is more possible to increase (i.e., $F$ is more possible to decrease).

\vspace{0.1cm}
\noindent
{\textbf{New wedge processing strategy.}}
Based on the above observation, we present a new wedge processing strategy: processing the wedges where the priorities of end-vertices are higher than the priorities of middle-vertices and start-vertices. We name the algorithm using this new strategy as \newaw. \newaw will perform more high-priority vertices as the end-vertices than \new because of the restriction of priorities of end-vertices. For example, considering the graph in Figure \ref{fig:ideas}(b), we have $p(v_0) > p(v_3) > p(u_0) > p(v_2) > p(v_1)$ according to their degrees. We analyse the processed wedges starting from $v_0$ to $v_3$, going through $u_0$. \new needs to process 5 wedges (i.e., $(v_0, u_0, v_1)$, $(v_0, u_0, v_2)$, $(v_0, u_0, v_3)$, $(v_3, u_0, v_1)$ and $(v_3, u_0, v_2)$) and 3 vertices (i.e., $v_1$, $v_2$ and $v_3$) are performed as end-vertices. Utilizing the new wedge processing strategy, in Figure \ref{fig:ideas}(b), the number of processed wedges of \newaw is still 5 (i.e., $(v_1, u_0, v_0)$, $(v_1, u_0, v_3)$, $(v_2, u_0, v_0)$, $(v_2, u_0, v_3)$ and $(v_3, u_0, v_0)$) but only 2 vertices with high-priorities (i.e., $v_0$ and $v_3$) are performed as end-vertices. Thus, the number of accessing different end-vertices is decreased from 3 to 2 (i.e., the accesses exhibit better locality). Also as shown in Figure \ref{fig:end}(b), after applying the new wedge processing strategy, the percentage of accesses of high-degree vertices (i.e., degree $> 2000$) increases from 9\% to 81\% on \texttt{Tracker} dataset.

\vspace{0.1cm}
\noindent
{\textbf{Time complexity unchanged.}}
Although the new wedge processing strategy can improve the CPU cache performance of \new, there are two questions arise naturally: (1) whether the number of processed wedges is still the same as \new; (2) whether the time complexity is still the same as \new after utilizing the new wedge processing strategy. We denote the set of processed wedges of \new as $W_{vp}$ and the set of processed wedges of \newaw as $W_{vp^+}$, we have the following lemma.

\begin{lemma}
\label{lemma:wedge}
$|W_{vp}| = |W_{vp^+}|$.
\end{lemma}

\begin{proof}
For a wedge $(u, v, w) \in W_{vp}$, it always satisfies $p(u) > p(v)$ and $p(u) > p(w)$ according to Algorithm \ref{algo:new}. For a wedge $(u, v, w) \in W_{vp^+}$, it always satisfies $p(w) > p(v)$ and $p(w) > p(u)$ according to the new wedge processing strategy. In addition, every vertex $u \in G$ has a unique $p(u)$ and the new wedge processing strategy does not change $p(u)$ of $u$. Thus, for each wedge $(u, v, w) \in W_{vp}$, we can always find a wedge $(w, v, u) \in W_{vp^+}$. Similarly, for each wedge $(u, v, w) \in W_{vp^+}$, we can always find a wedge $(w, v, u) \in W_{vp}$. Therefore, we prove that $|W_{vp}| = |W_{vp^+}|$.
\end{proof} 

Since no duplicate wedges are processed, based on the above lemma, \newaw will process the same number of wedges with \new. However, if only applying this strategy, when going through a middle-vertex, we need to check all its neighbors to find the end-vertices which have higher priorities than the middle vertex and the start-vertex. The time complexity will increase to $O(\sum_{u \in V(G), v \in \nb(u)}\degree(u)\degree(v))$ because each middle-vertex $v$ has $\degree(v)$ neighbors. In order to reduce the time complexity, for each vertex, we need to sort the neighbors in descending order of their priorities. After that, when dealing with a middle-vertex, we can early terminate the priority checking once we meet a neighbor which has a lower priority than the middle-vertex or the start-vertex. 

\subsection{Cache-aware Graph Projection}
\begin{figure}[!h]
\begin{centering}
\includegraphics[trim=0 10 10 -10,width=0.26\textwidth]{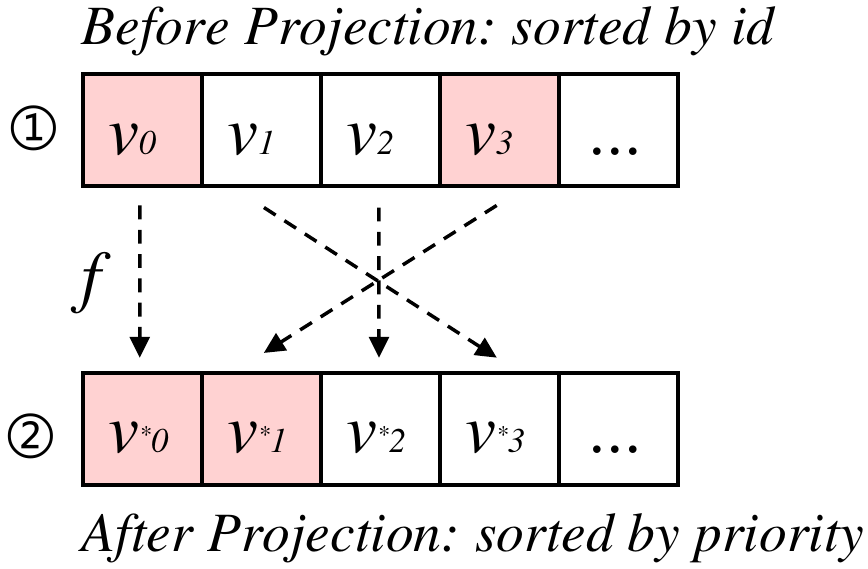}
\vspace*{-1mm}\caption{Illustrating the cache-aware graph projection}
\label{fig:projection}
\vspace*{-1mm}
\end{centering}
\end{figure}

\noindent
{\textbf{Motivation.}} After utilizing the cache-aware wedge processing strategy, end-vertices are mainly high-priority vertices. Generally, vertices are sorted by their ids when storing in the buffer. Figure \ref{fig:projection} shows accesses of the buffer when processing end-vertices (i.e., $v_0$ and $v_3$) starting from $v_0$ to $v_3$ and going through $u_0$ in Figure \ref{fig:ideas}(b) by \new. We can see that although end-vertices are mostly high-priority vertices, the distance between two end-vertices (e.g., $v_0$ and $v_3$) can be very long. This is because many low-priority vertices are stored in the middle of high-priority vertices. In addition, real graphs usually follow power-law distributions which do not contain too many vertices with high priorities (degrees). For example, in the \texttt{Tracker} dataset with about $40,000,000$ vertices, there are only $10,338$ vertices with degree $\geq 1000$, and only $1\%$  vertices ($400,000$) with degree $\geq 37$. Motivated by the above observations, we propose the graph projection strategy which can further improve the cache performance. 


\noindent
{\textbf{Graph projection strategy.}} The main idea of the graph projection strategy is projecting the given bipartite graph $G$ into a projection graph $G^*$ using a 1 to 1 bijective function $f$. The projection graph $G^*$ is defined as follows:

\begin{definition}[Projection Graph]Given a bipartite graph $G(V, E)$, a projection graph $G^*(V,E)$ is defined as:
$G^* \gets projection(G, f)$, where $f$ is a bijection from $E(G)$ to $E(G^*)$. For each $e = (u, v) \in E(G)$, $e^* = (u^*, v^*)= f(e)$ where $u^* \in U(G^*)$, $v^* \in L(G^*)$, and $u^*.id = rankU(u) + l$, $v^*.id = rankL(v)$. Here, $rankU(u)\in [0, r-1]$ ($rankL(v)\in [0, l-1]$) denotes the rank of the priority of $u \in U(G)$ (the rank of the priority of $v \in L(G)$).
\label{definition:projection}
\end{definition}

Unlike the conventional graph projection method in \cite{newman2001scientific1, newman2001scientific2} which projects a bipartite graph to a unipartite graph, our linear graph projection uses a 1 to 1 bijective function to relabel the vertex-IDs which does not change the graph structure. Thus, the number of vertices and edges are both unchanged after projecting. After projecting the original graph $G$ into the projection graph $G^*$, the vertices with high priorities will be stored together. In this manner, we can store more high-priority vertices consecutively in {\em HFA}. Figure \ref{fig:projection} illustrates the idea of graph projection using the example in Figure \ref{fig:ideas}(b).  After obtaining the projection graph $G^*$, we can see that the distance between two high-priority end-vertices becomes much shorter, e.g., the distance between $v^*_1$ and $v^*_2$ is 1 while the distance between $v_0$ and $v_3$ before projection is 3. In the experiments, we prove that the algorithms applying with the graph projection strategy achieves a much lower cache miss ratio than \new.

\subsection{Putting Cache-aware Strategies Together}
\label{sct:newaa}
\noindent
 {\textbf{The \newa algorithm.}} Putting the above strategies together, details of the algorithm \newa are shown in Algorithm \ref{algo:newa}. Given a bipartite graph $G$, \newa first generates a projection graph $\fg$ according to Definition \ref{definition:projection} and for each vertex $u^* \in V(\fg)$, we sort its neighbors. Then, \newa finds $\nbf(u^*)$ for each vertex $u^* \in V(\fg)$. For each vertex $v^* \in \nbf(u^*)$, we find $w^* \in \nbf(v^*)$ with $p(w^*) > p(u^*)$ and $p(w^*) > p(v^*)$ (lines 5 - 12).  After running lines 6 - 12, we get $|\nb(u^*) \cap \nb(w^*)|$ (i.e., $count\_wedge(w^*)$) for the start-vertex $u^*$ and the end-vertex $w^* \in \thop(u^*)$.  Finally, we compute $\btf_G$ (lines 13 - 15). 

\begin{algorithm}[th]
\small
\DontPrintSemicolon
\KwIn{$G(V = (U, L), E)$: the input bipartite graph}
\KwOut{$\btf_G$} 
$\btf_G \gets 0$\;
$\fg \gets projection(G, f)$ // Definition \ref{definition:projection} \;
compute $p(u^*)$ for each $u^* \in V(\fg)$  // Definition \ref{definition:priority} \;
sort $N(u^*)$ for each $u^* \in V(\fg)$ according to their priorities\;
\ForEach{$u^* \in V(\fg) $} {
    initialize hashmap $count\_wedge$ with zero\;
    \ForEach{$v^* \in \nbf(u^*)$} {
        \ForEach{$w^* \in \nbf(v^*): p(w^*) > p(u^*)$} {
            \If {$p(w^*) > p(v^*)$ } {
                $count\_wedge(w^*) \gets count\_wedge(w^*) + 1$\;
            }
            \Else {
                \Break\;
            }
        }
    }
    \ForEach{$w^* \in count\_wedge$} {
        \If {$count\_wedge(w^*) > 1$} {
            $\btf_G \gets \btf_G + \binom{count\_wedge(w^*)}{2}$\;
        }
    }
}
\Return{$\btf_G$}\; 
\caption{{\sc \newa}}
\label{algo:newa}
\end{algorithm}


\begin{theorem}
\label{theorem:newacor}
The \newa algorithm correctly solves the \btfc problem.
\end{theorem}

\begin{proof}
We prove that \newa correctly computes $\btf_G$ for a bipartite graph $G$. Since the graph projection strategy just renumbers the vertices, it does not affect the structure of $G$. Given a butterfly $[x, u, v, w]$, we assume $x$ has the highest priority. We only need to prove that \newa will count exactly once for each butterfly in Figure \ref{fig:proof}. Regarding the case in Figure \ref{fig:proof}(a), \ref{fig:proof}(b), or \ref{fig:proof}(c), \newa only counts the butterfly $[x, u, v, w]$ once from the wedges $(v, u, x)$ and $(v, w, x)$. Thus, we can get that the \newa algorithm correctly solves the \btfc problem.
\end{proof}

\begin{theorem}
\label{theorem:newatc}
The time complexity of \newa is $O(\sum_{(u, v) \in E(G)}min\{\degree(u), \degree(v)\})$.
\end{theorem}

\begin{proof}
The Algorithm \ref{algo:newa} has two phases including the initialization phase and $\btf_G$ computation phase. In the first phase, similar with \new, the algorithm needs $O(n + m)$ time to compute the priority number, sort the neighbors of vertices and compute the projection graph. Secondly, since we can use $O(1)$ time to process one wedge, we analyse the number of processed wedges by \newa as follows. In \newa, we only need to process the wedges where the degree of end-vertex is higher or equal than the middle-vertex. Considering an edge $(u, v) \in E(G)$ connecting an end-vertex $u$ and a middle-vertex $v$, we need to process $O(\degree(v))$ wedges containing $(u, v)$. Thus, we need to process $O(\sum_{(u, v) \in E(G)}min\{\degree(u), \degree(v)\})$ wedges in total. Therefore, the time complexity of \newa is $O(\sum_{(u, v) \in E(G)}min\{\degree(u), \degree(v)\})$.
\end{proof}


\begin{theorem}
\label{theorem:newasc}
The space complexity of \newa is $O(m)$.
\end{theorem}

\begin{proof}
This theorem is immediate.
\end{proof}

\noindent
{{\bf Remark.}} The cache-aware strategies proposed in this section are not applicable for the algorithms \bs and \bsa. This is because these strategies are priority-based strategies while the algorithms \bs and \bsa are not priority-based algorithms. 

\section{Extensions}
\label{sct:extension}

In this section, firstly, we extend our algorithms to compute $\btf_e$ for each edge $e$ in $G$. Secondly, we extend our algorithms to parallel algorithms. Thirdly, we introduce the external memory butterfly counting algorithm to handle large graphs with limited memory size.

\subsection{Counting the Butterflies for Each Edge}
\label{sct:edge}


Given an edge $e$ in $G$, we have the following equation \cite{wang2014rectangle}:

\vspace*{-1mm}
\begin{align}
\label{eq:edge}
& \btf_{e=(u, v)} = \sum_{w \in \thop(u), w \in \nb(v)}(|\nb(u) \cap \nb(w)| - 1)  \nonumber \\
& = \sum_{x \in \thop(v), x \in \nb(u)}(|\nb(v) \cap \nb(x)| - 1)\\ \nonumber
\end{align}
\vspace*{-7mm}

Based on the above equation, our \newa algorithm can be extended to compute $\btf_e$ for each edge $e$ in $G$. In Algorithm \ref{algo:newa}, for a start-vertex $u^*$ and a valid end-vertex $w^* \in \thop(u)$, the value $|\nb(u^*) \cap \nb(w^*)|$ is already computed which can be used directly to compute $\btf_e$.

Here, we present the \newae algorithm to compute $\btf_e$. The details of \newae are shown in Algorithm \ref{algo:newae}. In the initialization process, we initialize $\btf_e$ for each edge $e \in E(G)$. Then, for each start-vertex $u^*$, we run Algorithm \ref{algo:newa} Line 6 - Line 12 to compute $|\nb(u^*) \cap \nb(w^*)|$. After that, we run another round of wedge processing and update $\btf_{e(u, v)}$, $\btf_{e(v, w)}$ according to Equation \ref{eq:edge} (lines 5 - 14). Finally, we return $\btf_e$ for each edge $e$ in $G$.

In Algorithm \ref{algo:newae}, we only need an extra array to store $\btf_e$ for each edge $e$. In addition, because it just runs the wedge processing procedure twice, the time complexity of \newae is also $O(\sum_{(u, v) \in E(G)}min\{\degree(u), \degree(v)\})$.

\begin{algorithm}[th]
\small
\DontPrintSemicolon
\KwIn{$G(V = (U, L), E)$: the input bipartite graph}
\KwOut{$\btf_e$ for each $e \in E(G)$}

run Algorithm \ref{algo:newa} Line 2 - Line 4\;
$\btf_e \gets 0$ for each $e \in E(G)$\;

\ForEach{vertex $u^* \in V(\fg)$} {
    run Algorithm \ref{algo:newa} Line 6 - Line 12\;

    \ForEach{$v^* \in \nbf(u^*)$} {
        \ForEach{$w^* \in \nbf(v^*): p(w^*) > p(u^*)$} {
            \If {$p(w^*) > p(v^*)$ } {
                $\delta = count\_wedge(w) - 1$\;
                $(v, w) \gets \fin(v^*, w^*)$\;
                $(u, v) \gets \fin(u^*, v^*)$\;
                $\btf_{(u, v)} \gets \btf_{(u, v)} + \delta$\;
                $\btf_{(v, w)} \gets \btf_{(v, w)} + \delta$\;
            }
            \Else {
                \Break\;
            }
        }
    }
}
\Return{$\btf_e$ for each $e \in E(G)$}
\caption{{\sc \newae}}
\label{algo:newae}
\end{algorithm}


\subsection{Parallelization}
\label{sct:parallel}


\noindent
\textbf{Shared-memory parallelization.}
In Algorithm \ref{algo:newa}, only read operations occur on the graph structure. This motivates us to consider the shared-memory parallelization. Assume we have multiple threads and these threads can handle different start-vertices simultaneously. No conflict occurs when these threads read the graph structure simultaneously. However, conflicts may occur when they update $count\_wedge$ and $\btf_G$ simultaneously in Algorithm \ref{algo:newa}. Thus, we can divide the data-space into the global data-space and the local data-space. In the global data-space, the threads can access the graph structure simultaneously. In the local data-space, we generate $local\_count\_wedge$ and $local\_\btf_G$ for each thread to avoid conflicts. Thus, we can use $O(n * t + m)$ space to extend \newa into a parallel version, where $t$ is the number of threads.

\noindent
\textbf{Scheduling.}
In the parallel algorithm, we also need to consider the schedule strategies which may affect the load balance. We denote the workload for each start-vertex $u$ as $u.l$. We want to minimize the makespan $L$:

\begin{equation}
\label{equation:hit}
L = \max \limits_{1 \leq i \leq t}(\sum_{u \in V_i} u.l)\\
\end{equation}
Here, $V_i$ is the set of start-vertices on thread $i$.

The minimization of {\em L} is a well-known NP-hard optimization problem (i.e., the multiprocessor scheduling problem) \cite{garey2002computers} assuming we know the exact $u.l$ for each start-vertex $u$. Nevertheless, the workload is unknown prior, this makes the problem more challenging.

Below we discuss two types of schedule strategies: the dynamic scheduling and the static scheduling.

\noindent
{\em Dynamic scheduling.} For the dynamic scheduling, we queue all start-vertices in a specified order. Once a thread is idle, we dequeue a start-vertex and allocate it to the idle thread. The dynamic scheduling always delivers a schedule that has makespan at most $(2 - \frac{1}{t})L_{opt}$ where $L_{opt}$ is the optimal makespan even if the workload is unknown in advance \cite{graham1966bounds}. The bound is further reduced to $(\frac{4}{3} - \frac{1}{3t})L_{opt}$ by scheduling the job (i.e., start-vertex) with longer processing time (i.e., larger workload) first \cite{graham1969bounds}.
Since the vertex order may affect the performance, we compare three ordering strategies:

{\em 1) Heuristic strategy.} Sort start-vertices in non-ascending order by their estimated workloads (i.e., $\widetilde{u.l}$ for $u$);

{\em 2) Random strategy.} Sort start-vertices in random order.

{\em 3) Priority-based strategy.} Sort start-vertices in non-ascending order by their priorities.

To make our investigation more complete, we also consider the static scheduling. For the static scheduling, we need to pre-compute the allocations of start-vertices on the threads.

\noindent
{\em Static scheduling.} Here we discuss and compare three allocation strategies: 

{\em 1) Heuristic strategy.} We first estimate the workload for each start-vertex $u$ as: $\widetilde{u.l} = |S|$, where $S = \{w | w \in \nbf(v), v \in \nbf(u), p(w) > p(v) \}$. 
After that, we sort start-vertices by their estimated workloads in non-ascending order. Then we sequentially allocate these vertices and for each start-vertex, we always allocate it to the thread which has the minimal workload so far.

{\em 2) Random strategy.} For each start-vertex $u$, we randomly allocate $u$ to a thread under a uniform distribution.

{\em 3) Priority-based strategy.} First, we sort start-vertices according to their priorities in non-ascending order. Then, for each start-vertex $u$ with priority $p(u)$, we allocate $u$ to the thread $i$ if $p(u)\ mod\ t = i-1$.

\begin{figure}[htb]
\begin{centering}
\subfigure[\texttt{Tracker}, varying the number of threads]{
\begin{minipage}[b]{0.47\textwidth}
\includegraphics[trim=10 0 -10 0,clip,width=1\textwidth]{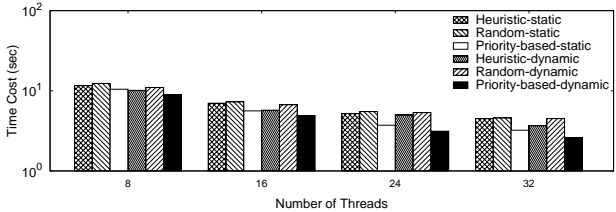}\vspace{-5mm}
\label{fig:t1}
\end{minipage}}
\subfigure[\texttt{Bi-twitter}, varying the number of threads]{
\begin{minipage}[b]{0.47\textwidth}
\vspace{-2mm}
\includegraphics[trim=10 0 -10 0,clip,width=1\textwidth]{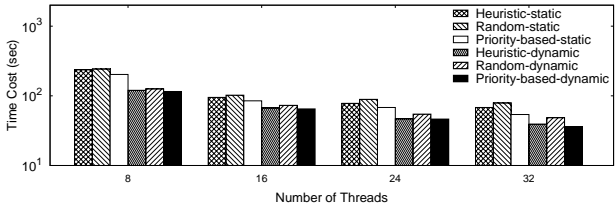}\vspace{-5mm}
\label{fig:t2}
\end{minipage}}
\vspace*{-4mm}\caption{Effect of scheduling}
\label{fig:tt}
\vspace*{-1mm}
\end{centering}
\end{figure}

In Figure \ref{fig:tt}, we compare the performance of the scheduling strategies by applying them into the parallel \newa algorithm on \texttt{Tracker} and \texttt{Bi-twitter} datasets. As shown in Figure \ref{fig:tt}, since the priority-based strategies are light-weight greedy strategies, they perform better than the other strategies in the static scheduling and the dynamic scheduling, respectively. Note that the heuristic strategies need additional pre-computation and the random strategies cannot achieve good performance. In addition, the priority-based dynamic scheduling strategy performs the best on these datasets. 


\begin{algorithm}[th]
\small
\DontPrintSemicolon
\KwIn{$G(V = (U, L), E)$: the input bipartite graph, $t$: number of threads}
\KwOut{$\btf_G$ for graph $G$}
run Algorithm \ref{algo:newa} Line 1 - Line 4\;
initialize $local\_count\_wedge[i]$ and $local\_\btf_G[i]$ for each thread $i\ \gets \ 1..t$\;
sort $u^* \in V(\fg)$ in non-ascending order by their priorities\;
    \ForEach{vertex $u^* \in V(\fg)$} {
        allocate $u^*$ to an idle thread $i$\;
            run Algorithm \ref{algo:newa} Line 6 - Line 15, replace $count\_wedge$, $\btf_G$ with $local\_count\_wedge[i]$, $local\_\btf_G[i]$\;
    }
/* on master thread */\;
$\btf_G \gets \btf_G + local\_\btf_G[i]$ for each thread $i\ \gets \ 1..t$\;

\Return{$\btf_G$}\;
\caption{{\sc \newa in parallel}}
\label{algo:newap}
\end{algorithm}

\noindent
\textbf{The algorithm \newa in parallel.} Since the priority-based dynamic scheduling strategy performs the best in the evaluation, we use it in our parallel algorithms. The details of the algorithm \newa in parallel are shown in Algorithm \ref{algo:newap}. Given a bipartite graph $G$, we first generate a projection graph $\fg$. Then, the algorithm sequentially processes the start-vertices in non-ascending order by their priorities. For a vertex  $u^* \in V(\fg)$, it will be dynamically allocated to an idle thread $i$. Note that, for each thread $i$, we generate an independent space for $local\_count\_wedge[i]$ and $local\_\btf_G[i]$. After all the threads finishing their computation, we compute $\btf_G$ on the master thread. 

\noindent
{\bf Remark.} The strategies presented here can be easily applied to the \bs, \bsa, and \new algorithms. 

\subsection{External memory butterfly counting}

In order to handle large graphs with limited memory size, we introduce the external memory algorithm \newaem in Algorithm \ref{algo:newaem} which is also based on the vertex priority. We first run an external sorting on the edges to group the edges with the same vertex-IDs together. Then we compute the priorities of vertices by sequentially scanning these edges once. Then, for each vertex $v \in V(G)$, we sequentially scan its neighbors from the disk and generate the wedges $(u, v, w)$ with $p(w) > p(v)$ and $p(w) > p(u)$ where $w \in \nb(v)$ and $u \in \nb(v)$ (lines 4 - 6). For each wedge $(u, v, w)$, we only store the vertex-pair $(u, w)$ on disk. After that, we maintain the vertex-pairs on disk such that all $(u, w)$ pairs with the same $u$ and $w$ values are stored continuously (line 7). This can be simply achieved by running an external sorting on these $(u, w)$ pairs. Then, we sequentially scan these vertex-pairs and for the vertex-pair $(u, w)$, we count the occurrence of it and compute $\btf_G$ based on Lemma \ref{lemma:existing} (lines 8 - 10).

\begin{algorithm}[th]
\small
\DontPrintSemicolon
\KwIn{$G(V = (U, L), E)$: the input bipartite graph}
\KwOut{$\btf_G$} 
sort all the edges $e \in G$ on disk\;
compute $p(u)$ for each $u \in V(G)$ on disk // Definition \ref{definition:priority} \;
$\btf_G \gets 0$\;
\ForEach{vertex $v \in G$} {
    \ForAll{$u, w \in \nb(v): p(w) > p(v), p(w) > p(u)$ by sequentially scanning $\nb(v)$ from disk} {
        store vertex-pair $(u, w)$ on disk\;
    }
}
sort all the vertex-pairs on disk\;
\ForEach{vertex-pair $(u, w)$} {
        $count\_pair(u, w) \gets$ count the occurrence of $(u, w)$ on disk sequentially\;
        $\btf_G \gets \btf_G + \binom{count\_pair(u, w)}{2}$\;
}
\Return{$\btf_G$}\;
\caption{{\sc \newaem}}
\label{algo:newaem}
\end{algorithm}

\noindent
\textbf{I/O complexity analysis.} We use the standard notations in \cite{aggarwal1988input} to analyse the I/O complexity of \newaem: $M$ is the main memory size and $B$ is the disk block size. The I/O complexity to scan $N$ elements is $scan(N) = \Theta(\frac{N}{B})$, and the I/O complexity to sort $N$ elements is $sort(N) = O(\frac{N}{B} log_{\frac{M}{B}}\frac{N}{B})$. In \newaem, the dominate cost is to scan and sort the vertex-pairs. Since there are $O(\sum_{(u, v) \in E(G)}min\{\degree(u), \degree(v)\})$ vertex-pairs generated by \newaem, the I/O complexity of \newaem is $O(scan(\sum_{(u, v) \in E(G)}min\{\degree(u), \degree(v)\}) + sort(\sum_{(u, v) \in E(G)}min\{\degree(u), \degree(v)\}))$.

\begin{table*}[ht]
\small
\caption{\label{table:datasets}
Summary of Datasets}
\vspace*{-7mm}
\begin{center}
\scalebox{0.888}{
\begin{tabular}{|c||c|c|c|c|c|c|c|c|c|}
\hline
\textbf{Dataset}& $|E|$  & $|U|$ & $|L|$  & $\btf_G$ & $\sum_{u \in L}d(u)^2$ & $\sum_{v \in R}d(v)^2$  & $\tc_{ibs}$ & $\tc_{new}$\\
\hline
DBPedia & 293{,}697 & 172{,}091 & 53{,}407& $3.76 \times 10^{6}$ & $6.30 \times 10^{5}$ & $2.46 \times 10^{8}$  &$6.30 \times 10^{5}$&$5.95\times 10^{5}$\\
\hline
Twitter& 1{,}890{,}661 & 175{,}214 & 530{,}418 & $2.07 \times 10^{8}$ & $7.42 \times 10^{7}$ & $1.94 \times 10^{9}$  &$7.42 \times 10^{7}$&$3.02 \times 10^{7}$\\
\hline
Amazon& 5{,}743{,}258 & 2{,}146{,}057 & 1{,}230{,}915 & $3.58 \times 10^{7}$ & $8.29 \times 10^{8}$ & $4.37 \times 10^{8}$  &$4.37 \times 10^{8}$&$6.90 \times 10^{7}$\\
\hline
Wiki-fr& 22{,}090{,}703  & 288{,}275 & 4{,}022{,}276   & $6.01 \times 10^{11}$ & $2.19 \times 10^{12}$ & $7.96 \times 10^{8}$  &$7.96 \times 10^{8}$&$7.08 \times 10^{7}$\\
\hline
Live-journal & $112{,}307{,}385$ & 3{,}201{,}203 & 7{,}489{,}073   & $3.30\times 10^{12}$ & $9.57 \times 10^{9}$ & $5.40 \times 10^{12}$ &$9.57 \times 10^{9}$&$8.01 \times 10^{9}$\\
\hline
Wiki-en  & 122{,}075{,}170 & 3{,}819{,}691 & 21{,}504{,}191   & $2.04\times 10^{12}$ & $1.26\times 10^{13}$ & $2.33\times 10^{10}$ &$2.33\times 10^{10}$&$9.32\times 10^{9}$\\
\hline
Delicious  & 101{,}798{,}957 & 833{,}081 & 33{,}778{,}221   & $5.69\times 10^{10}$ & $8.59\times 10^{10}$ & $5.28\times 10^{10}$ &$5.28\times 10^{10}$&$1.31\times 10^{10}$\\
\hline
Tracker  & 140{,}613{,}762 & 27{,}665{,}730 & 12{,}756{,}244   & $2.01\times 10^{13}$ & $1.73\times 10^{12}$ & $2.11\times 10^{14}$ &$1.73\times 10^{12}$&$7.83\times 10^{9}$\\
\hline
Orkut  & 327{,}037{,}487  & 2{,}783{,}196 & 8{,}730{,}857  & $2.21\times 10^{13}$ & $1.57\times 10^{11}$ & $4.90\times 10^{12}$ &$1.57\times 10^{11}$&$1.12\times 10^{11}$\\
\hline
Bi-twitter  & 601{,}734{,}937 & 20{,}826{,}115 & 20{,}826{,}110   & $6.30\times 10^{13}$ & $2.69\times 10^{13}$ & $3.48\times 10^{13}$ &$2.69\times 10^{13}$&$1.66\times 10^{11}$\\
\hline
Bi-sk & 910{,}924{,}634  & 25{,}318{,}075 & 25{,}318{,}075  & $1.22\times 10^{14}$ & $3.42\times 10^{13}$ & $1.80\times 10^{13}$ &$1.80\times 10^{13}$&$7.83\times 10^{10}$\\
\hline
Bi-uk & 1{,}327{,}632{,}357 & 38{,}870{,}511 & 38{,}870{,}511   & $4.89\times 10^{14}$ & $4.22\times 10^{13}$ & $4.16\times 10^{13}$ &$4.16\times 10^{13}$&$2.92\times 10^{11}$\\

\hline
\end{tabular}}
\end{center}
\vspace{-9mm}
\end{table*}

\section{Experiments}
\label{sct:experiment}
In this section, we present the results of empirical studies. In particular, our empirical studies have been conducted against the following algorithms:
1) the state-of-the art \bsa in \cite{sanei2018butterfly} as the baseline algorithm (we thank the authors for providing the code),
2) \new in Section \ref{sct:new}, 3) \newaw in Section \ref{sct:newaw}, 4) \newa in Section \ref{sct:newaa},
5) \bse, \newe, \newae by extending \bsa, \new and \newa, respectively, to compute $\btf_e$ for each edge $e$ in $G$,
6) the parallel version of \bsa, \new and \newa,
7) the most advanced approximate butterfly counting algorithm \bsap in \cite{sanei2018butterfly},
8) \newap by combining \newa with \bsap since \bsap relies on the exact butterfly counting techniques on samples,
and 9) the external memory algorithm \newaem.

The algorithms are implemented in C++ and the experiments are run on a Linux server with 2 $\times$ Intel Xeon E5-2698 processors and 512GB main memory.
Although most empirical studies have been against single core, we want our empirical studies to be conducted on the same computer as the evaluation of parallel performance.
{\em We terminate an algorithm if the running time is more than 10 hours}.

\subsection{Datasets}
\label{label:expsetting}
\noindent




\noindent
We use $12$ datasets in our experiments including all the $9$ real datasets in \cite{sanei2018butterfly} to ensure the fairness. We add $3$ more datasets to evaluate the scalability of our techniques.

The 9 real datasets we used are \texttt{DBPedia} \cite{DBPedia}, \texttt{Twitter} \cite{Twitter}, \texttt{Amazon} \cite{Amazon}, \texttt{Wiki-fr} \cite{Wiki-fr}, \texttt{Wiki-en} \cite{Wiki-en}, \texttt{Live-journal} \cite{Live-journal}, \texttt{Delicious} \cite{Delicious}, \texttt{Tracker} \cite{Tracker} and \texttt{Orkut} \cite{Orkut}.

To further test the scalability, we also evaluate three bipartite networks (i.e., \texttt{Bi-twitter}, \texttt{Bi-sk} and \texttt{Bi-uk}) which are sub-networks obtained from billion-scale real datasets (i.e., \texttt{twitter} \cite{Bi-twitter}, \texttt{sk-2005} \cite{Bi-sk} and \texttt{uk-2006-05} \cite{Bi-uk}).
In order to obtain bipartite-subgraphs from these two datasets, we put the vertices with odd ids in one group while the vertices with even ids in the other group and remove the edges which formed by two vertices with both odd ids or even ids.

The summary of datasets is shown in Table \ref{table:datasets}. $U$ and $L$ are vertex layers, $|E|$ is the number of edges.  $\btf_G$ is the number of butterflies. $\sum_{u \in L}d(u)^2$ and $\sum_{v \in R}d(v)^2$ represent the sum of degree squares for $L$ and $R$, respectively. $\tc_{ibs}$ is computed by $min\{\sum_{u \in L}d(u)^2, \sum_{v \in R}d(v)^2\}$ which is the time complexity bound of \bsa. $\tc_{new}$ is computed by $\sum_{(u, v) \in E(G)}min\{\degree(u), \degree(v)\}$ which is the time complexity bound of \new and \newa.


\subsection{Performance Evaluation}
In this section, we evaluate the performance of the proposed algorithms.


\begin{figure}[htb]
\begin{centering}
\includegraphics[trim=0 0 0 0,clip,width=0.46\textwidth]{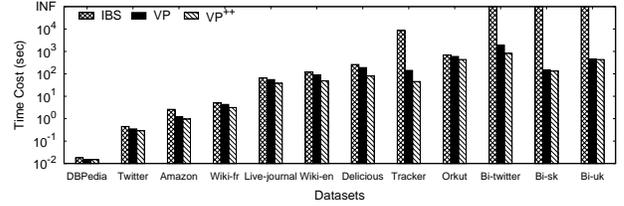}
\vspace*{-3mm}\caption{Performance on different datasets}
\label{fig:all}
\end{centering}
\end{figure}

\begin{figure}[htb]
\begin{centering}
\includegraphics[trim=0 0 0 0,clip,width=0.46\textwidth]{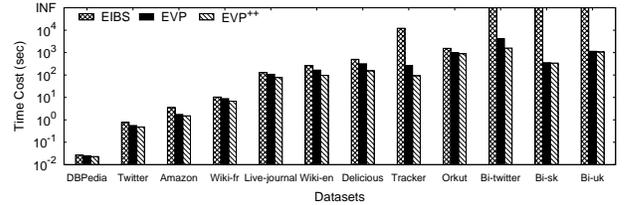}
\vspace*{-3mm}\caption{Performance on different datasets  (Counting the number of butterflies containing each edge $e$ in $G$)}
\label{fig:all_e}
\end{centering}
\end{figure}

\noindent
{\bf Evaluating the performance on all the datasets.}
In Figure \ref{fig:all}, we show the performance of the \bsa, \new and \newa algorithms on different datasets. We can observe that \newa is the most efficient algorithm, while \new also outperforms \bsa. This is because the \newa algorithm utilizes both the vertex-priority based optimization and the cache-aware strategies which significantly reduce the computation cost. On \texttt{Tracker}, the \new and \newa algorithms are at least two orders of magnitude faster than the \bsa algorithm. On \texttt{Bi-twitter}, \texttt{Bi-sk} and \texttt{Bi-uk}, the \bsa algorithm cannot finish within 10 hours. This is because the degree distribution of these datasets are skewed and high-degree vertices exist in both layers. For instance, $TC_{ibs}$ is more than 100$\times$ larger than $TC_{new}$ in \texttt{Tracker}. This property leads to a large number of wedge processing for \bsa while our \new and \newa algorithms can handle this situation efficiently. 

In Figure \ref{fig:all_e}, we show the performance of the algorithms which compute $\btf_e$ for each edge $e$ in $G$. The performance differences of these algorithms follow similar trends to those in Figure \ref{fig:all}. We can also observe that the \newae algorithm is the most efficient algorithm.


\begin{figure}[htb]
\begin{centering}
\includegraphics[trim=0 0 0 0,clip,width=0.46\textwidth]{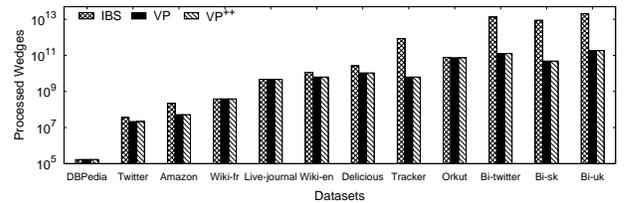}
\vspace*{-4mm}\caption{The number of processed wedges}
\label{fig:wedge}
\end{centering}
\end{figure}

\noindent
{\bf Evaluating the number of processed wedges.}
In Figure \ref{fig:wedge}, we show the number of processed wedges of the algorithms on all the datasets. We can observe that on \texttt{Tracker}, \texttt{Bi-twitter}, \texttt{Bi-sk} and \texttt{Bi-uk} datasets, \bsa needs to process $100 \times$ more wedges than \new and \newa. This is because $TC_{ibs}$ is much larger than $TC_{new}$ and there exist many hub-vertices in both $L$ and $R$ in these datasets. Thus, \new and \newa only need to process a limited number of wedges while \bsa should process numerous wedges no matter choosing which layer to start. We also observe that \new and \newa process the same number of wedges since \newa improves \new on cache performance which does not change the number of processed wedges.

\begin{figure}[htb]
\begin{centering}
\vspace*{-2mm}\subfigure[\texttt{Wiki-en}, varying $n$]{
\begin{minipage}[b]{0.2\textwidth}
\includegraphics[trim=0 0 0 0,clip,width=1\textwidth]{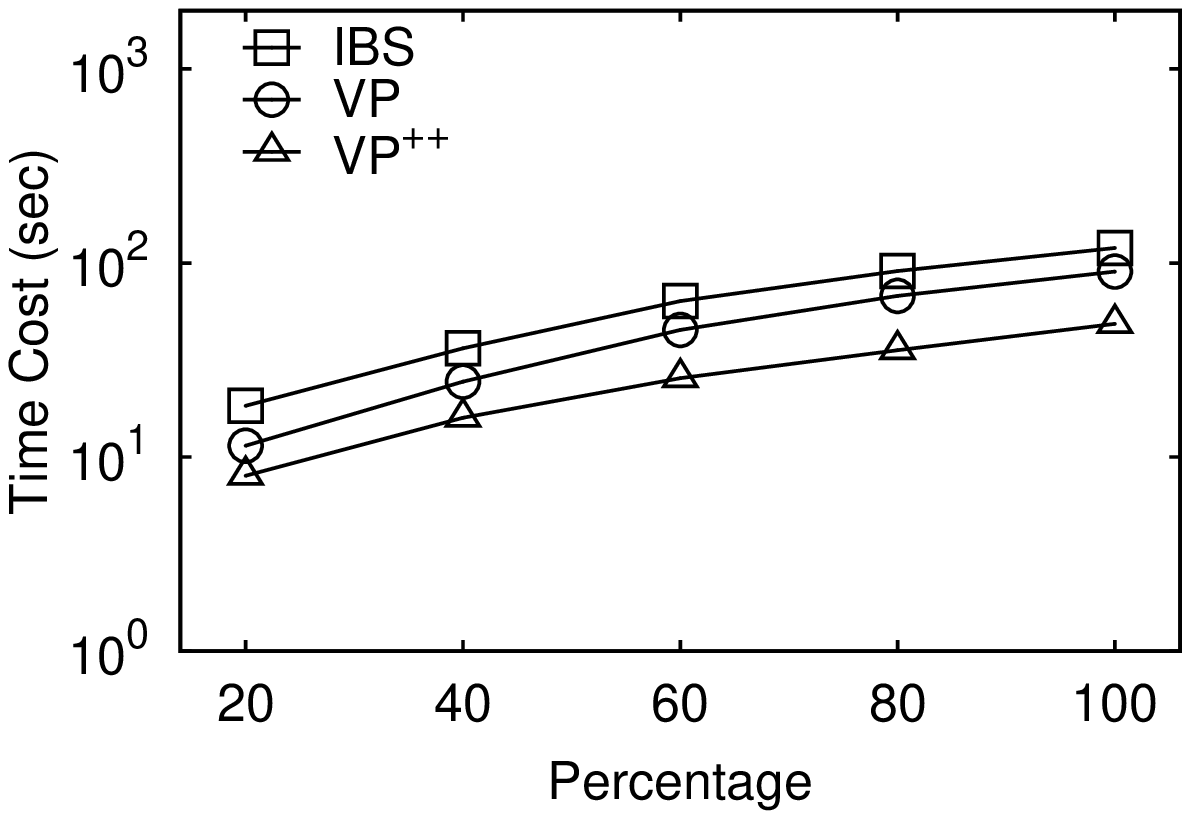}\vspace{-5.5mm}
\label{fig:n1}
\end{minipage}}
\vspace*{-4mm}\subfigure[\texttt{Delicious}, varying $n$]{
\begin{minipage}[b]{0.2\textwidth}
\includegraphics[trim=0 0 0 0,clip,width=1\textwidth]{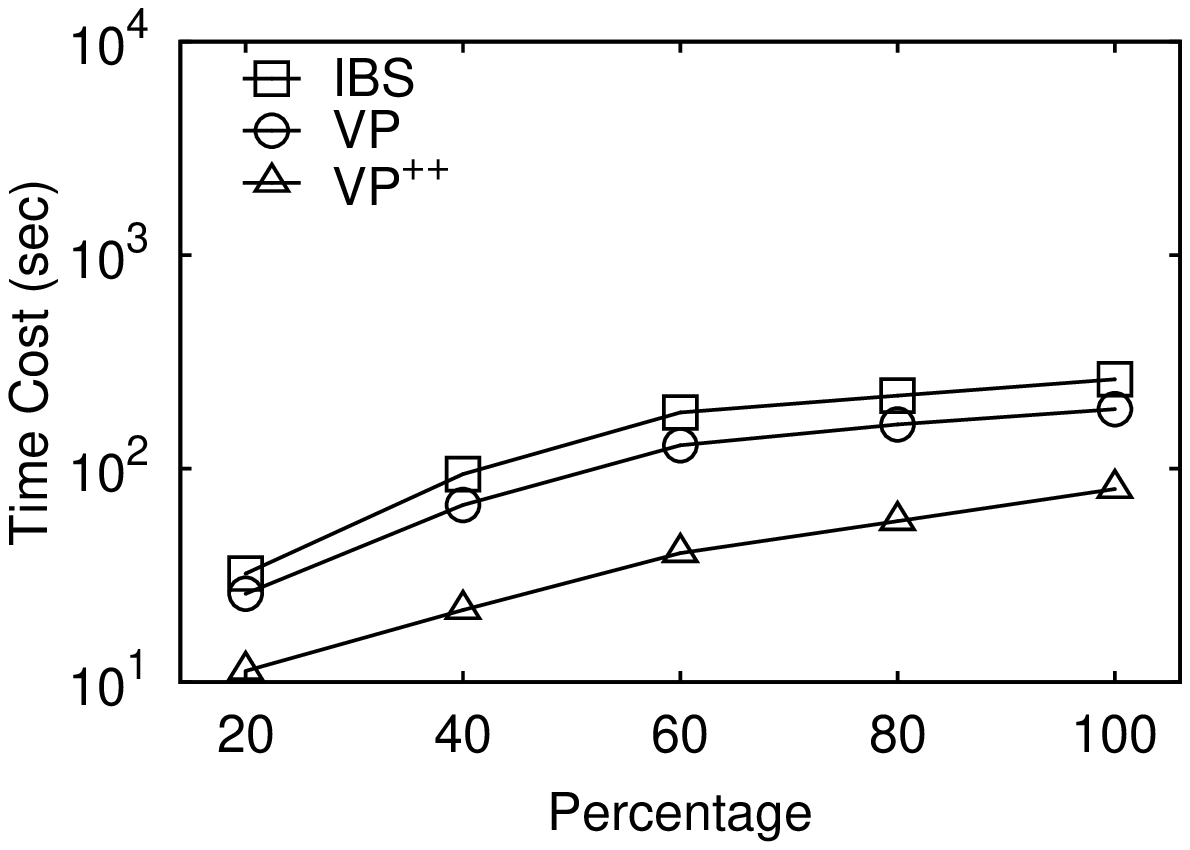}\vspace{-5.5mm}
\label{fig:n2}
\end{minipage}}
\subfigure[\texttt{Tracker}, varying $n$]{
\begin{minipage}[b]{0.2\textwidth}
\includegraphics[trim=0 0 0 0,clip,width=1\textwidth]{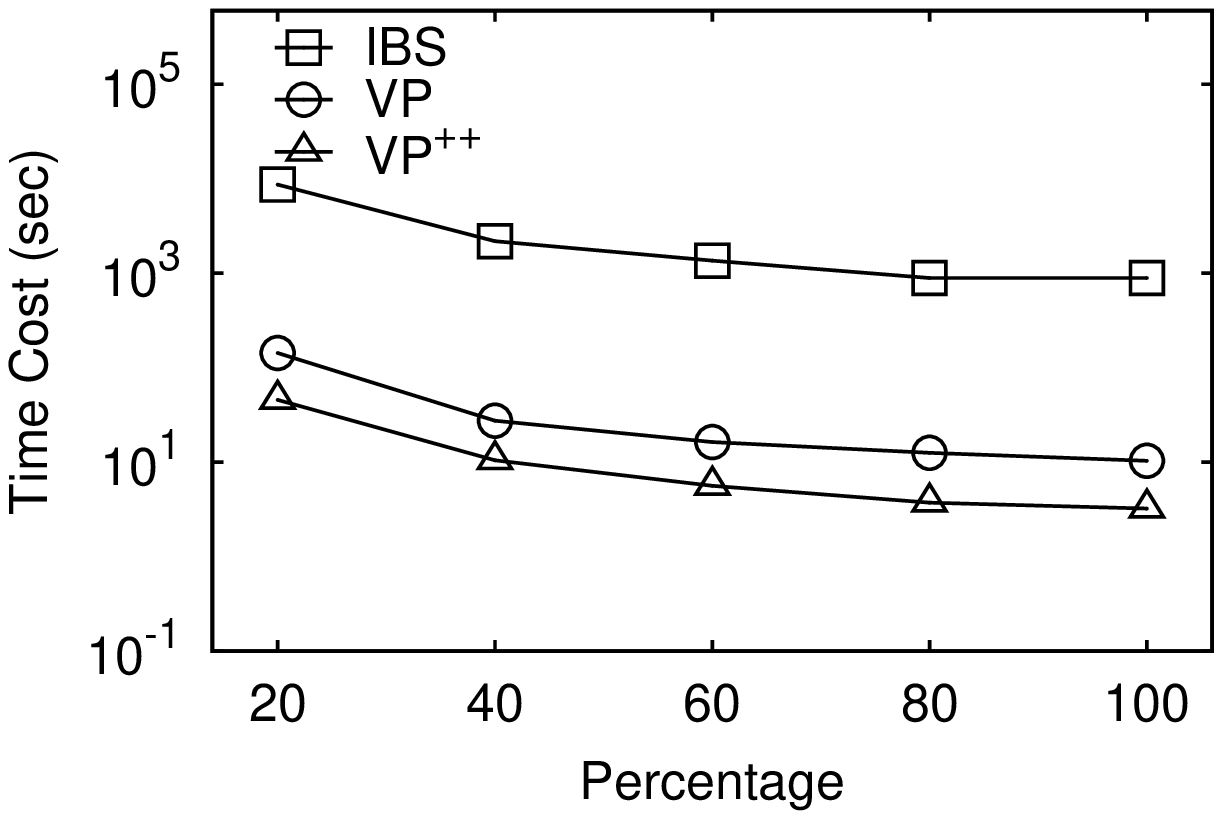}\vspace{-5.5mm}
\label{fig:n1}
\end{minipage}}
\subfigure[\texttt{Bi-twitter}, varying $n$]{
\begin{minipage}[b]{0.2\textwidth}
\includegraphics[trim=0 0 0 0,clip,width=1\textwidth]{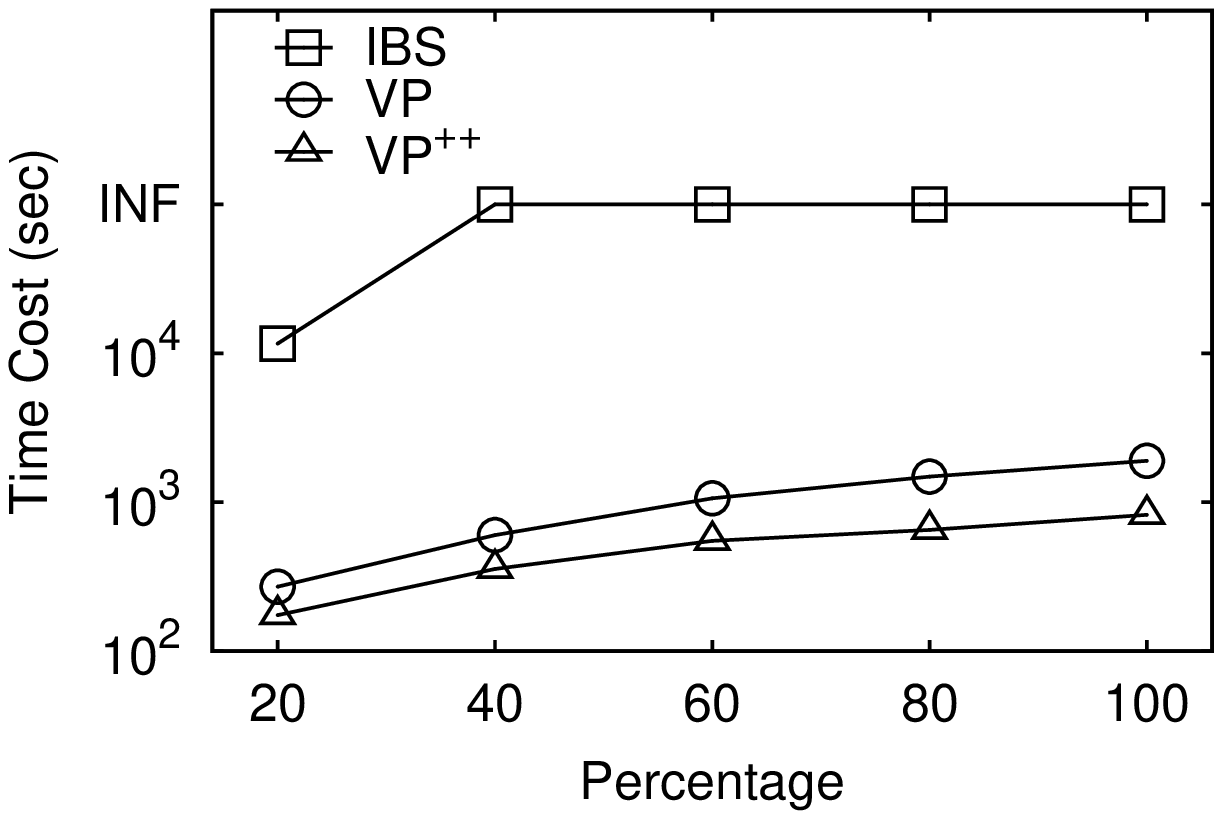}\vspace{-5.5mm}
\label{fig:n2}
\end{minipage}}
\vspace*{-4mm}\caption{Effect of graph size}
\label{fig:n}
\end{centering}
\end{figure}

\noindent
{\bf Scalability.} {\em Evaluating the effect of graph size.} Figure \ref{fig:n} studies the scalability of the algorithms by varying the graph size $n$ on four datasets. When varying $n$, we randomly sample 20\% to 100\% vertices of the original graphs, and construct the induced subgraphs using these vertices. We can observe that, \new and \newa are scalable and the computation cost of them all increases when the percentage of vertices increases. On \texttt{Bi-twitter}, \bsa can only complete when $n = 20\%$. As discussed before, \newa is the most efficient algorithm.

\begin{figure}[htb]
\begin{centering}
\vspace*{-3mm}\subfigure[\texttt{Wiki-en}, varying $t$]{
\begin{minipage}[b]{0.2\textwidth}
\includegraphics[trim=0 0 0 0,clip,width=1\textwidth]{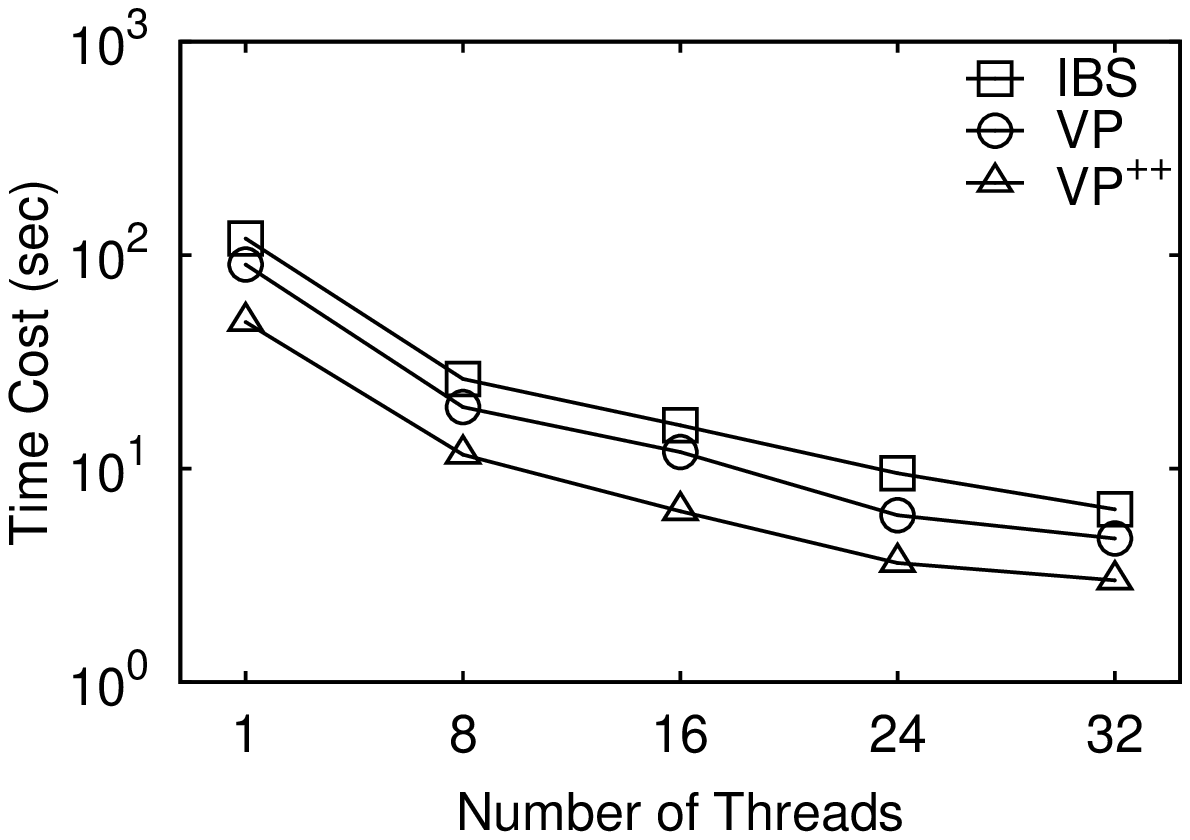}\vspace{-5.5mm}
\label{fig:n1}
\end{minipage}}
\vspace*{-4mm}\subfigure[\texttt{Delicious}, varying $t$]{
\begin{minipage}[b]{0.2\textwidth}
\includegraphics[trim=0 0 0 0,clip,width=1\textwidth]{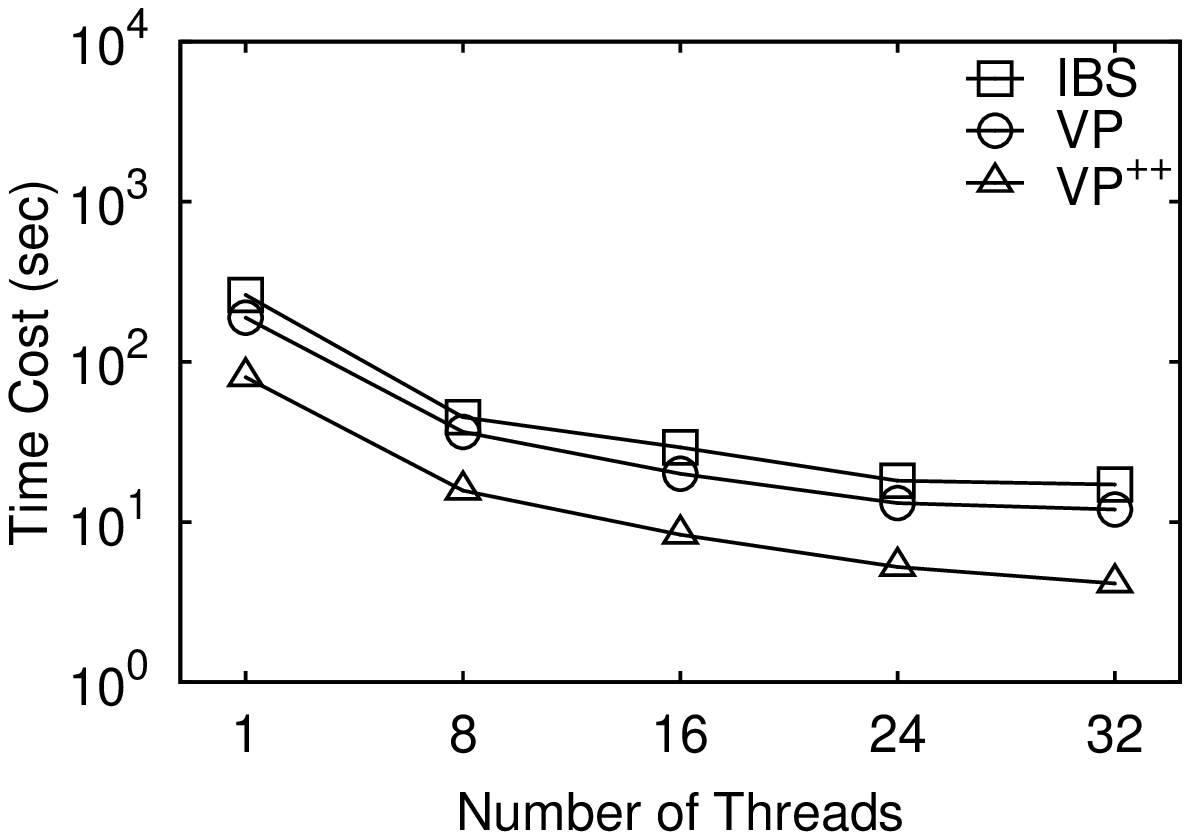}\vspace{-5.5mm}
\label{fig:n2}
\end{minipage}}
\subfigure[\texttt{Tracker}, varying $t$]{
\begin{minipage}[b]{0.2\textwidth}
\includegraphics[trim=0 0 0 0,clip,width=1\textwidth]{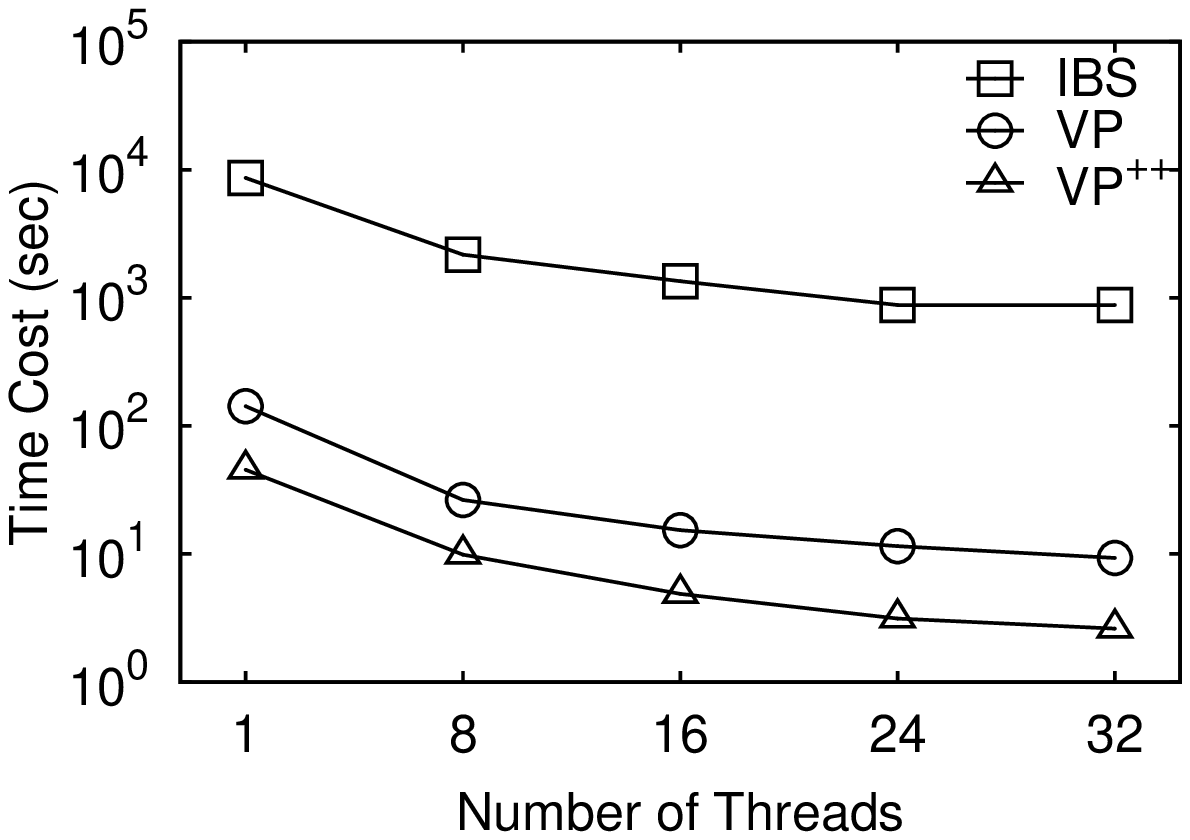}\vspace{-5.5mm}
\label{fig:n1}
\end{minipage}}
\subfigure[\texttt{Bi-twitter}, varying $t$]{
\begin{minipage}[b]{0.2\textwidth}
\includegraphics[trim=0 0 0 0,clip,width=1\textwidth]{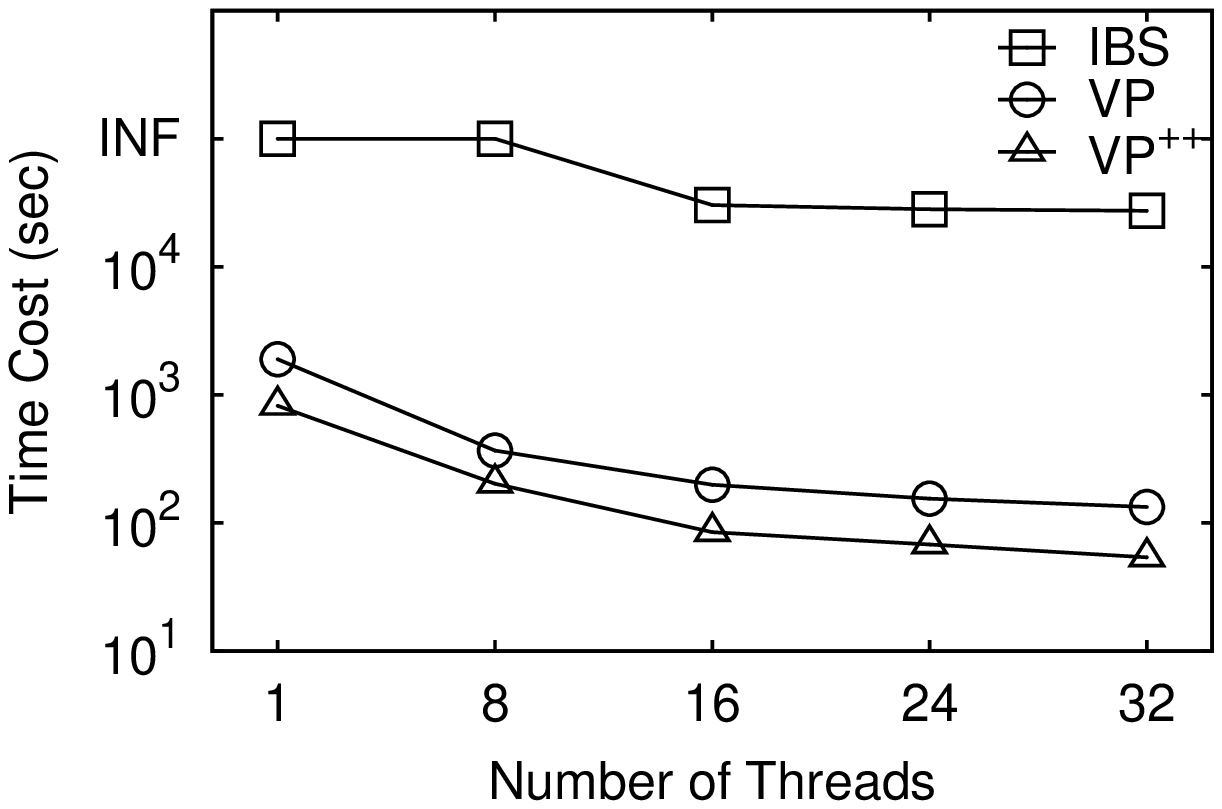}\vspace{-5.5mm}
\label{fig:n2}
\end{minipage}}
\vspace*{-4mm}\caption{Effect of $t$}
\label{fig:n}
\end{centering}
\end{figure}

\noindent
{\bf Speedup.} {\em Evaluating the parallelization.}
Figure \ref{fig:n} studies the performance of the \bsa, \new and \newa algorithms in parallel by varying the thread number $t$ from 1 to 32 on four datasets. The \bsa algorithm in parallel is not parallel-friendly. For example, on \texttt{Tracker}, the \bsa algorithm in parallel performs worse when $t$ increases from 16 to 32. On \texttt{Bi-twitter}, the algorithm \bsa in parallel cannot get a result within the timeout threshold when $t$ = 1 and $t$ = 8. We can also observe that, on all these datasets, the computation costs of the \new and \newa algorithms in parallel decrease when the number of threads increases and the algorithm \newa in parallel is more efficient than the other algorithms.


\begin{table}[!h]
\small
\caption{\label{table:miss-wiki}
Cache Statistics over \texttt{Wiki-en}}
\vspace{-6mm}
\begin{center}
\scalebox{0.85}{
\begin{tabular}{|c||c|c|c|c|}
\hline
\textbf{Algorithm} & $Cache$-$ref$ & $Cache$-$m$ & $Cache$-$mr$ & $Time(s)$\\
\hline
\textbf{$\new$} & $2.78 \times 10^{11}$ & $3.13 \times 10^{9}$ & 1.12\%& 90.41\\
\hline
\textbf{$\newac$} & $2.39 \times 10^{11}$ & $1.46 \times 10^9$ & 0.61\%&63.45\\
\hline
\textbf{$\newaw$} & $2.68 \times 10^{11}$ & $1.55 \times 10^{9}$ & 0.58\%& 65.26\\
\hline
\textbf{$\newa$} & $2.36 \times 10^{11}$ & $8.30 \times 10^{8}$ & 0.35\%& 48.60\\
\hline
\end{tabular}}
\end{center}
\vspace{-4mm}
\end{table}

\begin{table}[!h]
\small
\caption{\label{table:miss-deli}
Cache Statistics over \texttt{Delicious}}
\vspace{-6mm}
\begin{center}
\scalebox{0.85}{
\begin{tabular}{|c||c|c|c|c|}
\hline
\textbf{Algorithm} & $Cache$-$ref$ & $Cache$-$m$ & $Cache$-$mr$ & $Time(s)$\\
\hline
\textbf{$\new$} & $4.53 \times 10^{11}$ & $8.36 \times 10^{9}$ & 1.85\%&189.71\\
\hline
\textbf{$\newac$} & $4.19 \times 10^{11}$ & $4.08 \times 10^9$ & 0.97\%&133.48\\
\hline
\textbf{$\newaw$} & $4.40 \times 10^{11}$ & $3.87 \times 10^{9}$ & 0.88\%&102.82\\
\hline
\textbf{$\newa$} & $4.13 \times 10^{11}$ & $1.01 \times 10^{9}$ & 0.24\%&80.26\\
\hline
\end{tabular}}
\end{center}
\vspace{-4mm}
\end{table}

\begin{table}[!h]
\small
\caption{\label{table:miss-tracker}
Cache Statistics over \texttt{Tracker}}
\vspace{-6mm}
\begin{center}
\scalebox{0.85}{
\begin{tabular}{|c||c|c|c|c|}
\hline
\textbf{Algorithm} & $Cache$-$ref$ & $Cache$-$m$ & $Cache$-$mr$ & $Time(s)$\\
\hline
\textbf{$\new$} & $2.74 \times 10^{11}$ & $5.27 \times 10^{9}$ & 1.93\%&142.66\\
\hline
\textbf{$\newac$} & $2.40 \times 10^{11}$ & $1.88 \times 10^9$ & 0.84\%&87.61\\
\hline
\textbf{$\newaw$} & $2.52 \times 10^{11}$ & $1.75 \times 10^9$ & 0.78\%&82.16\\
\hline
\textbf{$\newa$} & $2.39 \times 10^{11}$ & $6.20 \times 10^8$ & 0.26\%&45.48\\
\hline
\end{tabular}}
\end{center}
\vspace{-4mm}
\end{table}

\begin{table}[!h]
\small
\caption{\label{table:miss-twitter}
Cache Statistics over \texttt{Bi-twitter}}
\vspace{-6mm}
\begin{center}
\scalebox{0.85}{
\begin{tabular}{|c||c|c|c|c|}
\hline
\textbf{Algorithm} & $Cache$-$ref$ & $Cache$-$m$ & $Cache$-$mr$ & $Time(s)$\\
\hline
\textbf{$\new$} & $4.87 \times 10^{12}$ & $4.96 \times 10^{10}$ & 1.02\%&1897.15\\
\hline
\textbf{$\newac$} & $4.55 \times 10^{11}$ & $2.47 \times 10^{10}$ & 0.54\%&1261.11\\
\hline
\textbf{$\newaw$} & $4.58 \times 10^{12}$ & $2.39 \times 10^{10}$ & 0.52\%&1096.86\\
\hline
\textbf{$\newa$} & $4.54 \times 10^{12}$ & $1.35 \times 10^{10}$ & 0.30\%&822.31\\
\hline
\end{tabular}}
\end{center}
\vspace{-4mm}
\end{table}

\noindent
{\bf Evaluating the cache-aware strategies.}
In Table \ref{table:miss-wiki}, Table \ref{table:miss-deli}, Table \ref{table:miss-tracker} and Table \ref{table:miss-twitter}, we evaluate the cache-aware strategies on \texttt{Wiki-en}, \texttt{Delicious}, \texttt{Tracker} and \texttt{Bi-twitter}, respectively. Here, \emph{Cache-ref} denotes the total cache access number. \emph{Cache-m} denotes the total cache miss number which means the number of cache references missed. \emph{Cache-mr} denotes the percentage of cache references missed over the total cache access number. \emph{Time} denotes the computation time of different algorithms. Here, $\newaw$ is the $\new$ algorithm deploying with only the cache-aware wedge processing strategy. $\newac$ is the $\new$ algorithm deploying with only the graph projection strategy. \new has the largest number of cache-miss on all the datasets. By utilizing the cache-aware wedge processing, compared with \new, \newaw reduces the number of cache miss over 50\% on all the datasets. By utilizing the cache-aware projection, compared with \new, \newac also reduces over 50\% cache miss on all the datasets.  \newa achieves the smallest cache-miss-numbers, and reduces the cache-miss-ratio significantly on all these datasets since \newa combines the two cache-aware strategies together. Compared with \new, \newa reduces over more than 70\% cache miss on all the testing datasets. 

\begin{figure}[htb]
\begin{centering}
\subfigure[\texttt{Tracker}, varying $p$]{
\begin{minipage}[b]{0.2\textwidth}
\includegraphics[trim=0 0 0 0,clip,width=1\textwidth]{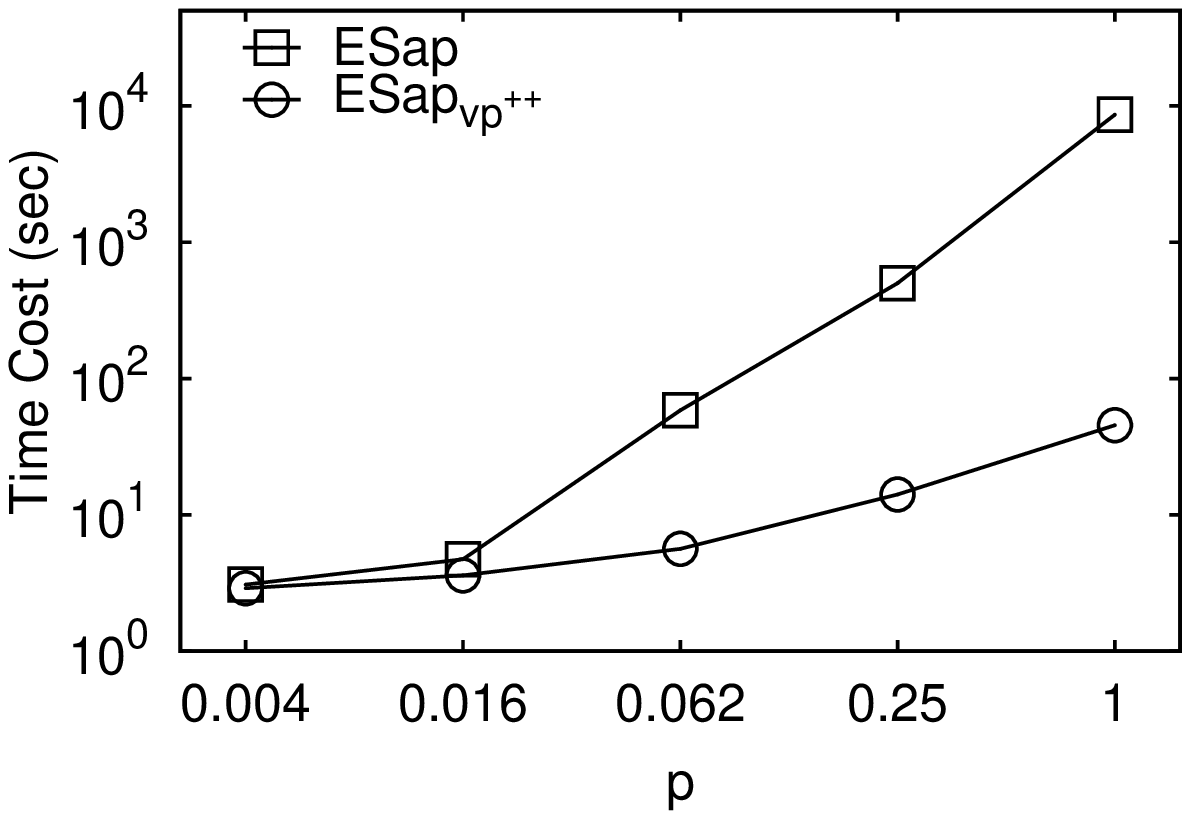}\vspace{-5.5mm}
\label{fig:p3}
\end{minipage}}
\subfigure[\texttt{Bi-twitter}, varying $p$]{
\begin{minipage}[b]{0.2\textwidth}
\includegraphics[trim=0 0 0 0,clip,width=1\textwidth]{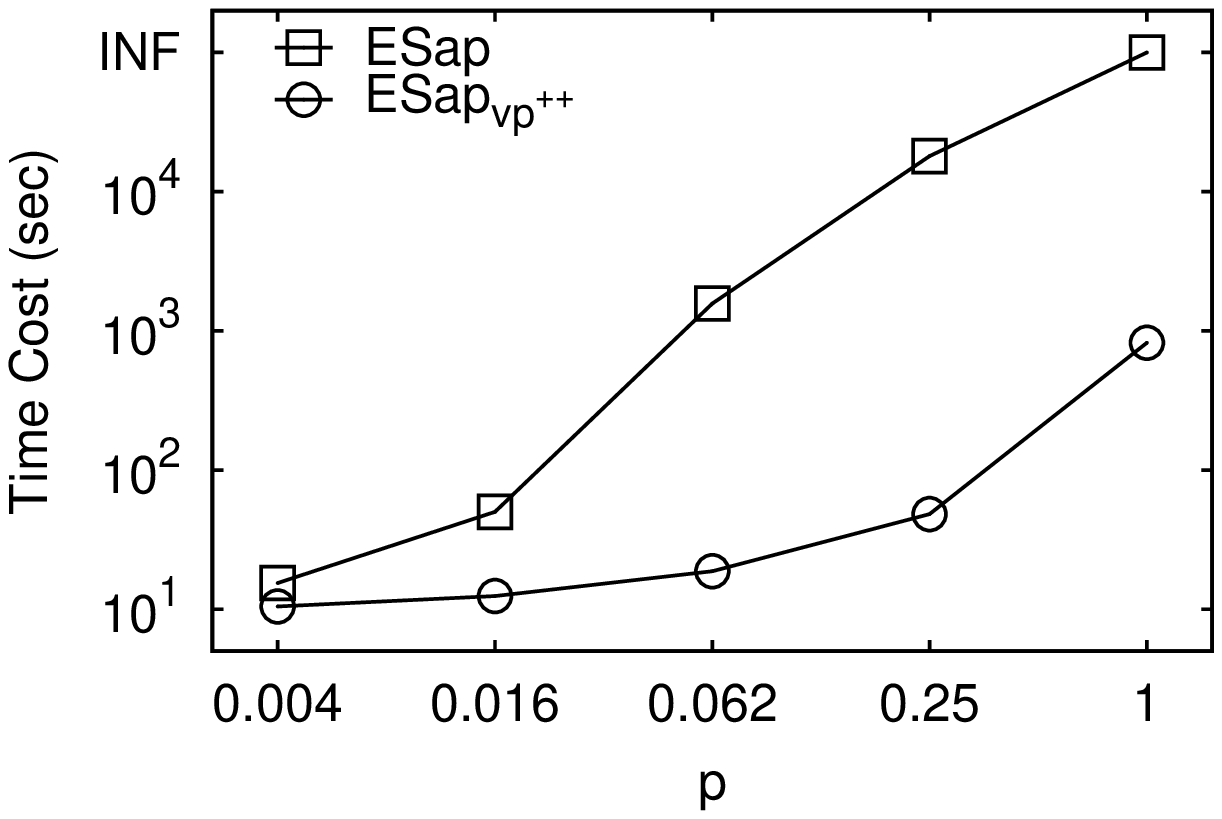}\vspace{-5.5mm}
\label{fig:p4}
\end{minipage}}
\vspace*{-4mm}\caption{Effect of $p$}
\label{fig:p}
\end{centering}
\end{figure}

\begin{figure}[htb]
\begin{centering}
\subfigure[\texttt{Tracker}, varying $\epsilon$]{
\begin{minipage}[b]{0.2\textwidth}
\includegraphics[trim=0 0 0 0,clip,width=1\textwidth]{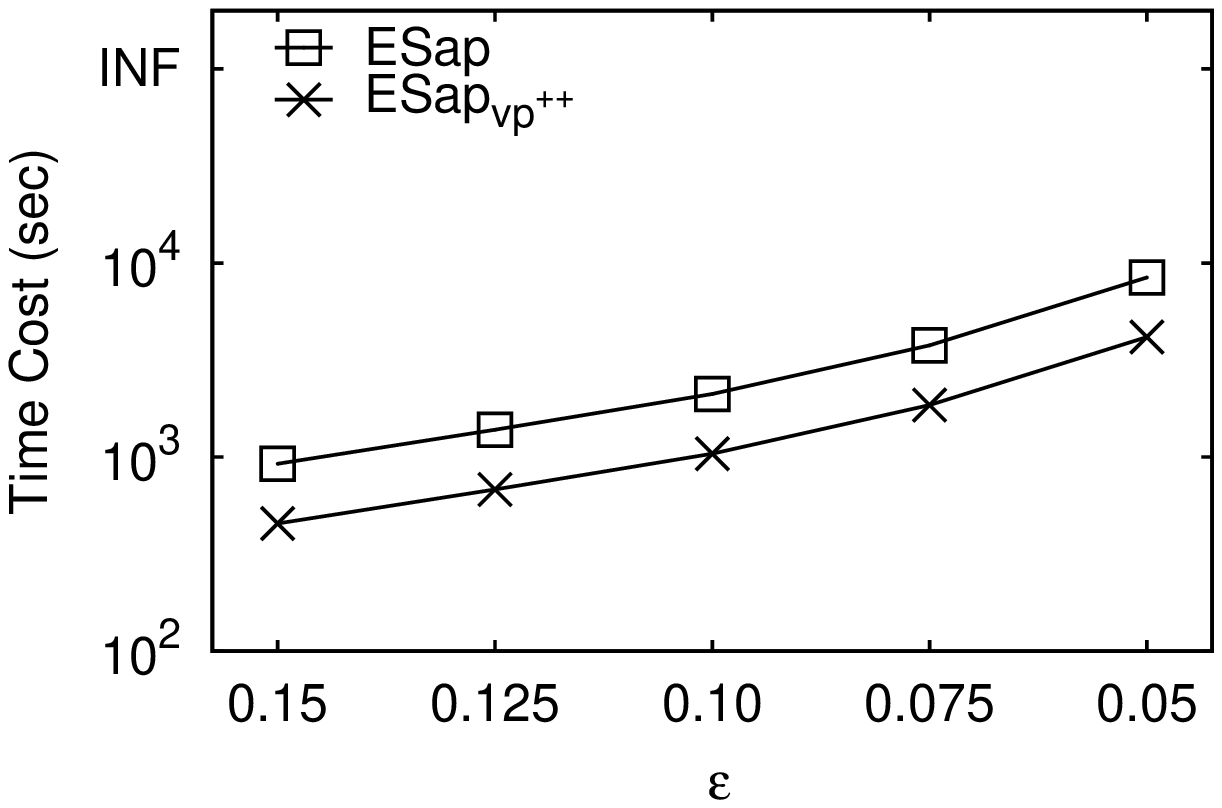}\vspace{-5.5mm}
\label{fig:p3}
\end{minipage}}
\subfigure[\texttt{Bi-twitter}, varying $\epsilon$]{
\begin{minipage}[b]{0.2\textwidth}
\includegraphics[trim=0 0 0 0,clip,width=1\textwidth]{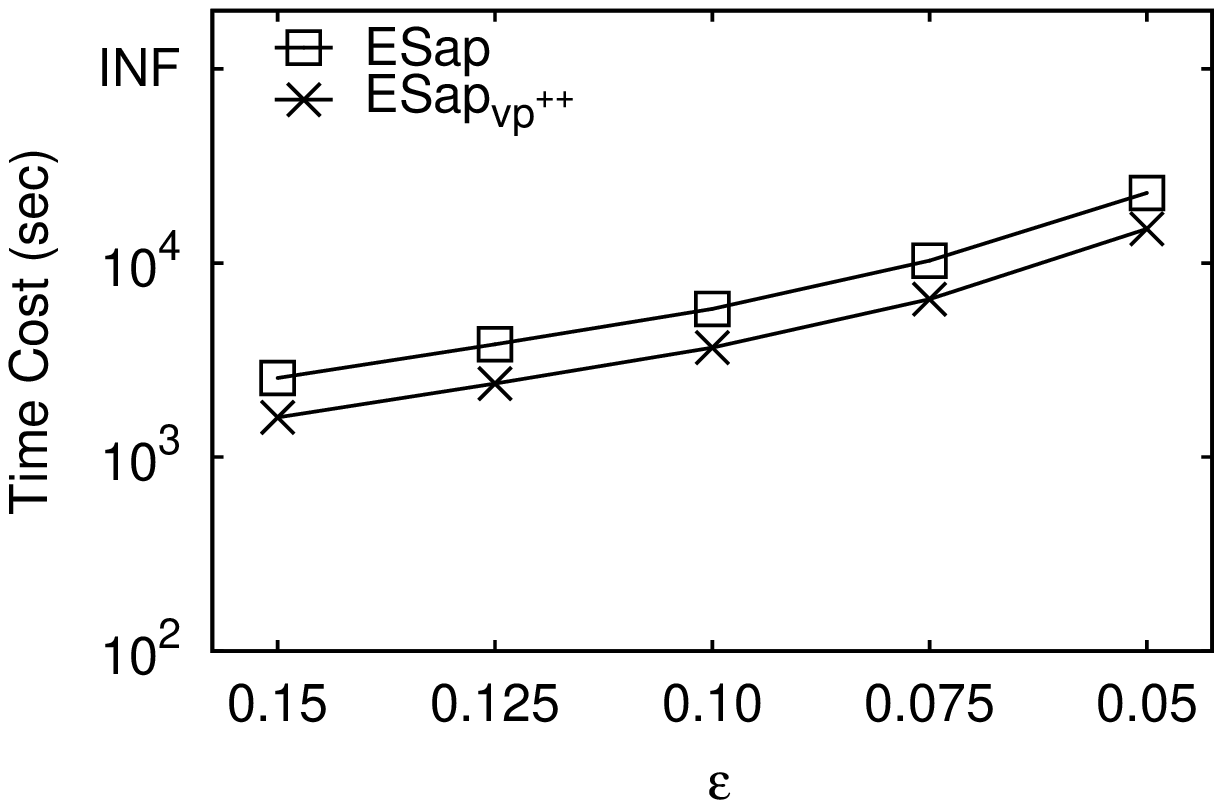}\vspace{-5.5mm}
\label{fig:p4}
\end{minipage}}
\vspace*{-4mm}\caption{ Effect of $\epsilon$}
\label{fig:e}
\end{centering}
\end{figure}

\noindent
{\bf Speeding up the approximate butterfly counting algorithm.}
In the approximate algorithm \bsap \cite{sanei2018butterfly}, the exact butterfly counting algorithm \bsa is served as a basic block to count the butterfly exactly in a sampled graph. Since \newa and \bsa both count the number of butterfies exactly, the approximate algorithm \newap can be obtained by applying \newa in \bsap without changing the theoretical guarantee.

In Figure \ref{fig:p}, we first evaluate the average running time of \bsap and \newap for each iteration by varying the probability $p$. Comparing two approximate algorithms, \newap outperforms \bsap under all the setting of $p$ on \texttt{Tracker} and \texttt{Bi-twitter} datasets. Especially, on these two datasets, \newap is more than one order of magnitude faster than \bsap when $p \geq 0.062$.

In the second experiment, we run the algorithms to yield the theoretical guarantee $Pr[|\hat\btf_G-\btf_G|>\epsilon \btf_G] \leq \delta$ as shown in \cite{sanei2018butterfly}. We vary $\epsilon$ and fix $\delta = 0.1$ and $p$ as the optimal $p$ as suggested in \cite{sanei2018butterfly}. We can see that for these two approximate algorithms, the time costs are increased on these two datasets in order to get a better accuracy and \newap significantly outperforms \bsap as mentioned before.

\begin{figure}[htb]
\begin{centering}
\subfigure[Time Cost, varying $n$]{
\begin{minipage}[b]{0.2\textwidth}
\includegraphics[trim=0 0 0 0,clip,width=1\textwidth]{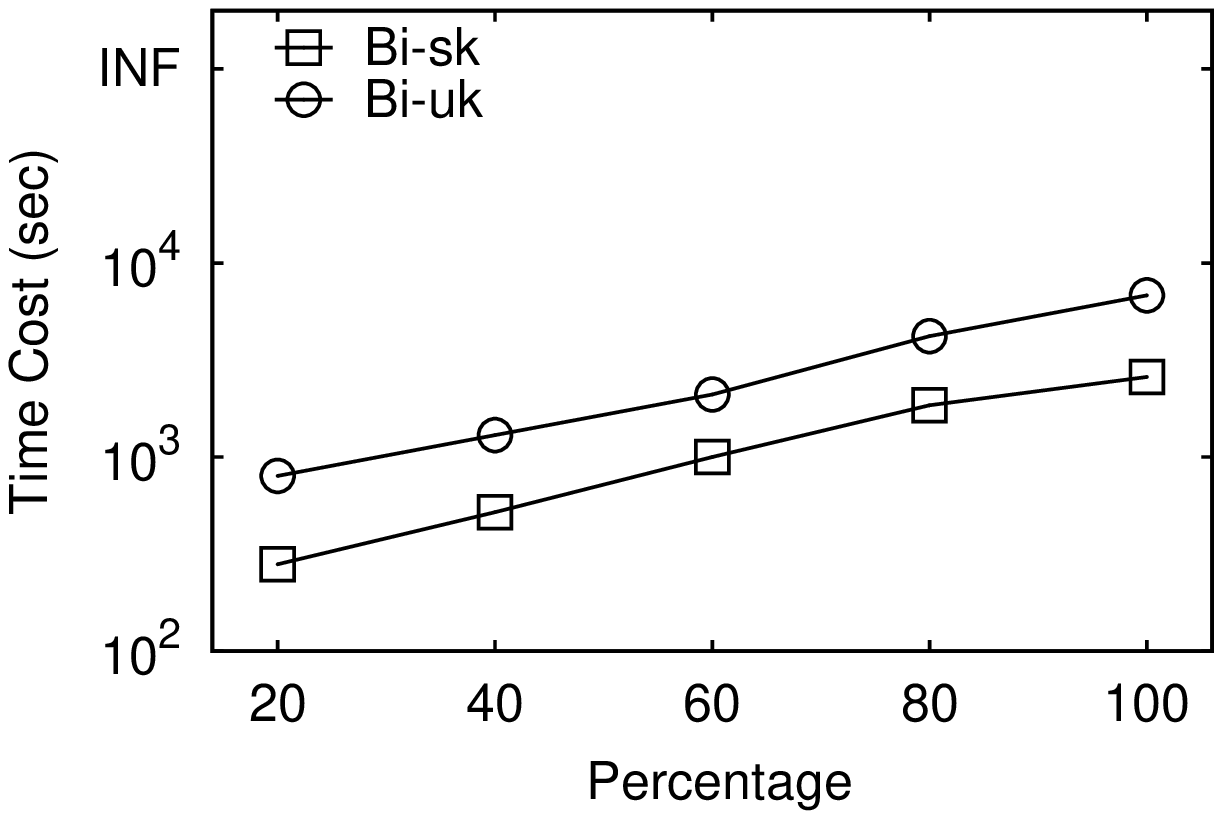}\vspace{-5.5mm}
\label{fig:p3}
\end{minipage}}
\subfigure[I/O, varying $n$]{
\begin{minipage}[b]{0.2\textwidth}
\includegraphics[trim=0 0 0 0,clip,width=1\textwidth]{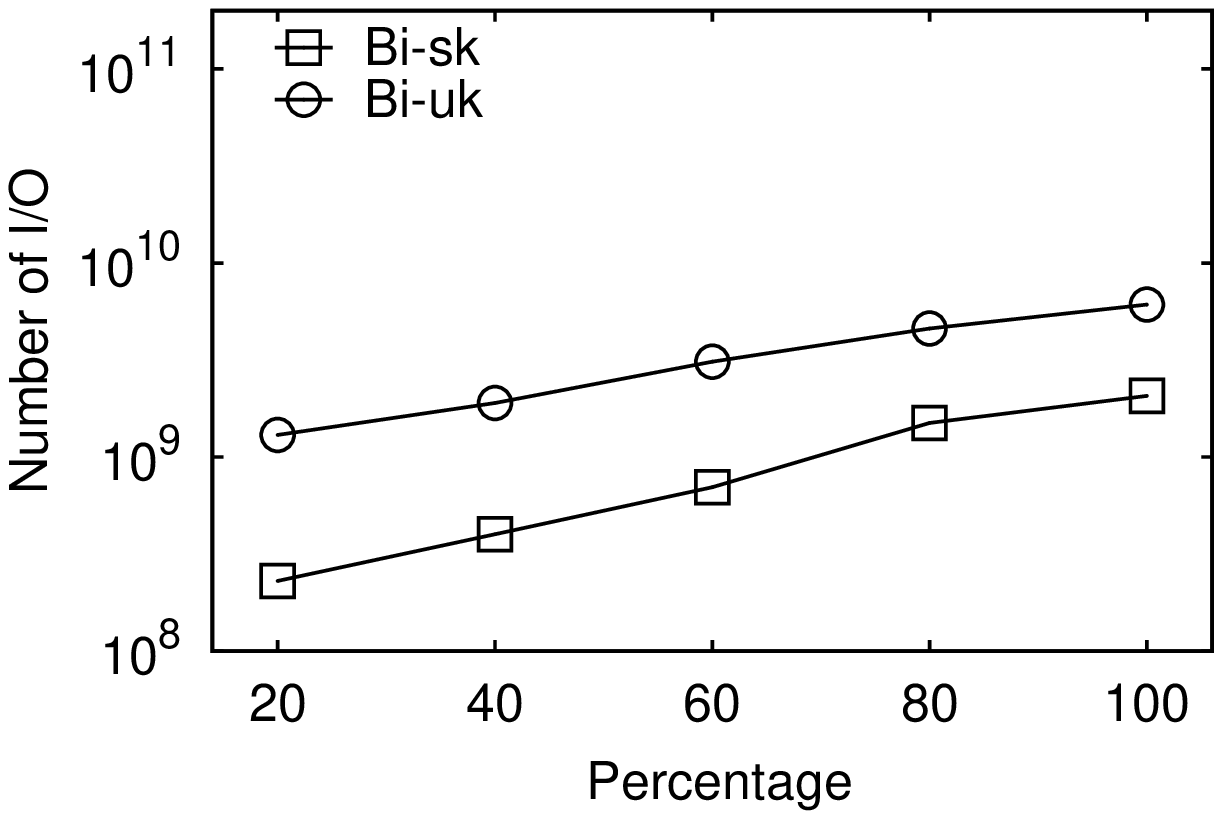}\vspace{-5.5mm}
\label{fig:p4}
\end{minipage}}
\vspace*{-4mm}\caption{ Evaluating the external memory algorithm}
\label{fig:em}
\end{centering}
\end{figure}

\noindent
{\bf Evaluating the external memory algorithm.}
In Figure \ref{fig:em}, we evaluate the scalability of the external memory algorithm \newaem on two large datasets \texttt{Bi-sk} and \texttt{Bi-uk} by varying the graph size $n$. We limit the memory size to 1GB in our evaluation. On \texttt{Bi-sk} and \texttt{Bi-uk}, we can see that the time cost and I/O both increase with the percentage of vertices increases.

\begin{table}[htb]
\small
\caption{Time cost compared with Gorder}
\vspace*{-7mm}
\begin{center}
\scalebox{0.74}{
\begin{tabular}{|c||c|c|c|c|c|c|}
\hline
 \multirow{2}*{\textbf{Dataset}}& \multicolumn{2}{c|}{$Renumbering\ time$} & \multicolumn{2}{c|}{$Computation\ time$} &\multicolumn{2}{c|}{$Total\ time$}\\
\cline{2-7}
 ~  & Projection & Gorder & Projection & Gorder & Projection & Gorder \\
\hline
DBPedia & \textbf{0.01}&	0.04	& \textbf{0.02}	&0.03	&\textbf{0.03}	&0.07\\
\hline
Twitter & \textbf{0.06}&	4.26	& 0.29	&\textbf{0.25}	&\textbf{0.35}	&4.51\\
\hline
Amazon & \textbf{0.30}&	3.56	& \textbf{0.96}	&1.46	&\textbf{1.26}	&5.02\\
\hline
Wiki-fr  & \textbf{0.49}	&28.51	&\textbf{3.16}	& 5.28	&\textbf{3.65}	&33.79\\
\hline
Live-journal & \textbf{1.32}	&125.96	& \textbf{37.86}	&52.76	&\textbf{39.18}	&178.72\\
\hline
Wiki-en  & \textbf{3.02}	&856.07	& \textbf{48.60}	&75.78	&\textbf{51.62}	&931.85\\
\hline
Delicious  & \textbf{3.82}	&2225.44	& \textbf{80.26}	&134.86	&\textbf{84.08}	&2360.30\\
\hline
Tracker  & \textbf{4.89}	&315.01	& \textbf{45.48}	&56.13	&\textbf{50.37}	&371.13\\ 			
\hline
Orkut  & \textbf{2.17}	&1615.01	& \textbf{435.12}	&553.03	&\textbf{437.29}	&2168.04\\			
\hline
Bi-twitter  & \textbf{6.64}	&3211.63	& \textbf{822.31}	&1276.63	&\textbf{828.95}	&4488.26\\			
\hline
Bi-sk & \textbf{8.32}	&605.87	&133.34	& \textbf{107.07}	&\textbf{141.66}	&692.94\\ 			
\hline
Bi-uk & \textbf{9.91}	&1231.93	&435.29	& \textbf{401.64}	&\textbf{445.20}	&1633.57\\
\hline
\end{tabular}
}
\vspace{-5.3mm}
\end{center}
\label{table:go-time}
\end{table}

\begin{table}[htb]
\small
\caption{Cache statistics compared with Gorder}
\vspace*{-7mm}
\begin{center}
\scalebox{0.686}{
\begin{tabular}{|c||c|c|c|c|c|c|}
\hline
 \multirow{2}*{\textbf{Dataset}}& \multicolumn{2}{c|}{$Cache\ reference$} & \multicolumn{2}{c|}{$Cache\ miss$} &\multicolumn{2}{c|}{$Cache\ miss\ ratio$}\\
\cline{2-7}
 ~  & Projection & Gorder & Projection & Gorder & Projection & Gorder \\
\hline
DBPedia & \textbf{4.02 $\times 10^7$}&	$5.61 \times 10^7$&	\textbf{4.54 $\times 10^4$} &	$1.20 \times 10^5$	&\textbf{0.11\%}&	0.21\%\\
\hline
Twitter & \textbf{8.89 $\times 10^8$}&	$9.56 \times 10^8$&	5.09 $\times 10^5$ & \textbf{4.68 $\times 10^5$}	&0.06\%&	\textbf{0.05}\%\\
\hline
Amazon & \textbf{2.51 $\times 10^9$}&	$2.52 \times 10^9$&	\textbf{8.93 $\times 10^6$} &	$1.02 \times 10^7$	&\textbf{0.36\%}&	0.40\%\\
\hline
Wiki-fr  & \textbf{1.34$\times 10^{10}$}	&1.38$\times 10^{10}$&	\textbf{7.33$\times 10^{7}$}	&8.40$\times 10^{7}$&	\textbf{0.55\%}&	0.61\%\\
\hline
Live-journal & 1.72$\times 10^{11}$	&\textbf{1.68$\times 10^{11}$}&	\textbf{6.68$\times 10^{8}$}	&8.02$\times 10^{8}$		&\textbf{0.39\%}&	0.48\%\\
\hline
Wiki-en  & 2.36$\times 10^{11}$&	\textbf{2.30$\times 10^{11}$}	&\textbf{8.30$\times 10^{8}$}	&	1.29$\times 10^{9}$		&\textbf{0.35\%}&	0.56\%\\
\hline
Delicious  & 4.13$\times 10^{11}$&	\textbf{4.03$\times 10^{11}$}&	\textbf{1.01$\times 10^{9}$}		&1.63$\times 10^{9}$	&	\textbf{0.24\%}	&0.40\%\\
\hline
Tracker  & 2.39$\times 10^{11}$	&\textbf{2.34$\times 10^{11}$}	&\textbf{6.20$\times 10^{8}$}		&7.29$\times 10^{9}$		&\textbf{0.26\%}	&0.31\%\\
\hline
Orkut  & 2.69$\times 10^{12}$	&\textbf{2.58$\times 10^{12}$}	&\textbf{7.21$\times 10^{9}$}&	8.38$\times 10^{9}$&	\textbf{0.27\%}	&0.33\%\\
\hline
Bi-twitter  & 4.54$\times 10^{12}$	&\textbf{4.49$\times 10^{12}$}	&\textbf{1.35$\times 10^{10}$}	&3.04$\times 10^{10}$	&\textbf{0.30\%}	&0.68\%\\
\hline
Bi-sk & 1.64$\times 10^{12}$&	\textbf{1.58$\times 10^{12}$}&	2.29$\times 10^{9}$&	\textbf{2.01$\times 10^{9}$}&	0.14\%&	\textbf{0.13\%}\\
\hline
Bi-uk & 6.15$\times 10^{12}$&	\textbf{6.00$\times 10^{12}$}&	3.67$\times 10^{9}$&	\textbf{3.21$\times 10^{9}$}&	0.06\%&	\textbf{0.05\%}\\
\hline
\end{tabular}
}
\end{center}
\vspace{-3mm}
\label{table:go-cache}
\end{table}

\noindent
{\bf Graph projection vs Gorder.}
In \cite{wei2016speedup}, the authors proposed the {\em Gorder} model to reduce the cache miss in graph algorithms. Here, we replace the graph projection with Gorder in \newa and evaluate the difference of performances. 

Table \ref{table:go-time} shows the time cost. We can observe that the renumbering time cost of the graph projection is much less than Gorder on all datasets. This is because graph projection can be simply obtained according to the priority number of vertices while Gorder needs complex renumbering computation. Regarding the computation time, the performance of the algorithm with graph projection is better than the algorithm with Gorder on 9 datasets while the algorithm with Gorder is better on 3 datasets. Finally, the total time cost of graph projection is better than Gorder.


Table \ref{table:go-cache} shows the cache statistics. Firstly, they have a similar number of cache references since the renumbering process does not change the algorithm itself. Secondly, graph projection achieves a better CPU performance than Gorder on almost all the datasets (i.e., less cache misses and less cache miss ratios on 9 datasets) when dealing with the butterfly counting problem with the \newa algorithm.

In summary, our graph projection strategy is more suitable when dealing with the butterfly counting problem.

\vspace{-0.1cm}
\section{Conclusion}

\label{sct:conclusion}
In this paper, we study the \btfc problem. We propose a \newf\ \new which can effectively handle high-degree vertices. We also propose the \newaf\ \newa which improves the CPU cache performance of \new with two cache-aware strategies. We conduct extensive experiments on real datasets and the result shows that our \newa algorithm significantly outperforms the state-of-the-art algorithms. In the future, we plan to study the butterfly counting problem in a distributed environment \cite{wang2010mapdupreducer, lin1993data, lai2015scalable} or a data stream \cite{wang2015ap, henzinger1998computing}.
\balance
\newpage

{
\bibliographystyle{abbrv}
\bibliography{paper}
}

\end{document}